\documentclass{article}

\usepackage[english]{babel}
\usepackage{bbold}
\usepackage{graphicx}
\usepackage{caption}
\usepackage{subcaption}
\usepackage{array}
\usepackage{booktabs}
\usepackage{array, makecell}
\usepackage{multirow}
\usepackage{cite}
\usepackage[letterpaper,top=2cm,bottom=2cm,left=3cm,right=3cm,marginparwidth=1.75cm]{geometry}
\usepackage{amsfonts}
\usepackage{dsfont}
\usepackage{subcaption}
\usepackage{mathrsfs}
\usepackage{amsmath}
\usepackage{amsthm}
    \newtheorem{theorem}{Theorem}
    \newtheorem{corollary}{Corollary}[theorem]
    
    \newtheorem{definition}{Definition}
\usepackage{graphicx}
\usepackage[colorlinks=true, allcolors=blue]{hyperref}
\usepackage[normalem]{ulem}
\usepackage{comment}
\usepackage{booktabs} 
\usepackage{array} 
\usepackage{{makecell}}
    \setcellgapes{3pt}\makegapedcells
\usepackage{tabularx}
\usepackage{cancel}

\newcolumntype{M}{>{\hfil$\displaystyle}c<{$\hfil}} 
\newcolumntype{m}{>{$\displaystyle}X<{$}}
\newcolumntype{L}{>{$\displaystyle}l<{$}} 
\newcommand{\FracArgument}[2]{\left(\frac{#1}{#2}\right)}
\newcommand{\e}{\text{e}}

\title{{\bf The role of conjugacy in the dynamics of time of arrival operators}}
\author{Dean Alvin L. Pablico\footnote{dlpablico@up.edu.ph}, John Jaykel P. Magadan \footnote{jpmagadan@up.edu.ph},\\
Carl Anthony L. Arguelles \footnote{clarguelles@up.edu.ph} and Eric A. Galapon\footnote{eagalapon@up.edu.ph}\\Theoretical Physics Group, National Institute of Physics\\University of the Philippines, Diliman, 1101 Philippines}

\begin{document}
\maketitle

\begin{abstract}
The construction of time of arrival (TOA) operators canonically conjugate to the system Hamiltonian entails finding the solution of a specific second-order partial differential equation called the time kernel equation (TKE). In this paper, we provide an exact analytic solution of the TKE for a special class of potentials satisfying a specific separability condition. The solution enables us to investigate the time evolution of the eigenfunctions of the conjugacy-preserving TOA operators (CPTOA) and show that they exhibit unitary arrival at the intended arrival point at a time equal to their corresponding eigenvalues. We also compare the dynamics between the TOA operators constructed by quantization and those independent of quantization for specific interaction potentials. We find that the CPTOA operator possesses smoother and sharper unitary dynamics over the Weyl-quantized one within numerical accuracy. 
\end{abstract}

\section{Introduction}

The incorporation of time as a quantum dynamical observable has long been a subject of intense debate since the inception of quantum theory. Unlike position, momentum, and energy observables, the absence of a well-defined time operator in standard quantum mechanics represents a significant theoretical gap.  This highlights the chief weakness of the theory \cite{vonneuman,Hilgevoord2005}, since experimental measurements of time can be routinely carried out in laboratories using quantum clocks  that are intricately tied to physical systems undergoing change \cite{Muga2000}.

Perhaps the most scrutinized avenue for incorporating time observables into the quantum framework is the concept of time of arrival (TOA) for elementary particles \cite{Muga2000,Egusquiza2000,Baute2000,Muga2008,Muga2009,Galapon2018,Galapon2012,Galapon2009,Galapon2001,Galapon2002,Galapon2002a,Galapon2004,Galapon2004a,Galapon2005,Galapon2005a,Galapon2006,Galapon2008,Galapon2009a,Sombillo2012,Sombillo2014,Sombillo2016,Sombillo2018,PhysRevA.109.012216,Jurman2021,PhysRevA.103.012225,Sombillo2023,J.Leon20000,muga1995brouard,sokolovskibaskin,halliwel2009,PhysRevA.58.840,wangxiong,Jadczyk}. This area has garnered increased attention in recent years, especially in light of recent strides in attosecond strong-field physics, enabling time of arrival measurements accurate to the attosecond regime ($10^{-18} \mathrm{s}$) \cite{Eckle2008,Eckle2008a,S.Landsman2015,Hofmann2019,Sainadh2019,Kheifets_2020}. In theoretical discourse, the customary way of defining the problem is by considering a structureless elementary particle of mass $\mu$ and Hamiltonian $H(q,p)=p^2/2\mu+V(q)$ in one dimension for some interaction potential $V(q)$. Given that the particle is initially located at some point $(q,p)$ in phase space at time $t=0$, the central question arises: How do we determine the statistics of measured arrival times at a designated arrival point, say at the origin $q = 0$? 

Several approaches have been introduced in the literature to resolve the quantum time of arrival problem. They can be generally categorized into two: (i) non-time-operator-based formulations and (ii) time-operator-based formulations. The first category includes the construction of the relevant TOA probability distributions within the context of decoherent histories \cite{Egusquiza2000, halliwel2009}, Wigner's crossing states \cite{Baute2000}, stochastic processes \cite{PhysRevA.109.012216}, current density \cite{muga1995brouard}, path integrals \cite{sokolovskibaskin}, and Bohmian mechanics \cite{PhysRevA.58.840}, among others. The second category relies on constructing relevant time operators through methods such as quantization \cite{Galapon2018}, supraquantization \cite{Galapon2012,Pablico2023}, the Mandelstam-Tamm method \cite{Jadczyk}, or using the energy shift generator \cite{wangxiong}, among others. For a more extensive review on the quantum TOA problem, we refer readers to Refs. \cite{Muga2000,Muga2008,Muga2009}.

Now, we argue that if one truly wishes to incorporate time as a dynamical observable and elevate its status on the same footing as the position, momentum, and energy observables, then one has to confront the problem with time operators. In Ref. \cite{Galapon2018}, two of us have considered the problem of constructing TOA operators by canonical quantization methods. There, quantization is performed on the expansion of the classical time of arrival expression given by
\begin{equation}\label{classicaltoa}
\mathcal{T}_0(q,p)=-\mathrm{sgn}(p)\sqrt{\frac{\mu}{2}}\int_{0}^{q}\frac{dq'}{\sqrt{H(q,p)-V(q')}},
\end{equation}
about the free time of arrival, where $\mathrm{sgn}(p)$ is the signum function. The constructed TOA operators are formulated as integral operators of the form 
\begin{equation}\label{toaopformq}
(\hat{\mathrm{T}}_Q\varphi)(q)=\frac{\mu}{i\hbar}\int_{-\infty}^{\infty} \,dq' \,\mathrm{sgn}(q-q')\,T_Q(q,q')\, \varphi(q'),
\end{equation}
in coordinate representation, where $T_Q(q,q')$ is called the time kernel factor (TKF). The exact form of $T_Q(q,q')$ depends explicitly on  the chosen quantization rule $Q$. For instance, the time kernel factor corresponding to the Weyl quantization $(W)$ of the classical arrival time appears as
\begin{equation}\label{tweyl}
T_W(q,q')=\frac{1}{4}\int_{0}^{q+q'} ds \,\, {}_0F_{1}\left(;1;\left(\frac{\mu}{2\hbar^2}\right)(q-q')^2\left[V\left(\frac{q+q'}{2}\right)-V\left(\frac{s}{2}\right)\right]\right).
\end{equation}

For Weyl, symmetric, and Born-Jordan ordering rules, the corresponding quantized TOA operators exhibit the following key properties: Hermiticity, time-reversal symmetry, and adherence to the quantum-classical correspondence in the classical limit as $\hbar \to 0$. These operators also demonstrate a dynamic characteristic: their eigenfunctions unitarily evolve according to Schr\"{o}dinger equation in such a way that the corresponding probabilities of finding the arriving particle in the neighborhood of the arrival point is maximum at their respective eigenvalues. This property establishes a fundamental link between TOA measurements and the appearance of the incident particle at the arrival point \cite{Galapon2004a,Galapon2009a,Sombillo2016}. As a result, it is this unitary arrival property that we should expect and demand for a physically meaningful TOA operator \cite{Galapon2018}.

Nevertheless, it is known that the quantized TOA operators generally fail to satisfy the required conjugacy with the system Hamiltonian, $	[\hat{\mathrm{H}},\hat{\mathrm{T}}]=i\hbar \mathbb{1}$, for arbitrary potentials due to the known obstructions to quantization \cite{Galapon2001,Gotay1999,Groenewold1946}. While all quantizations lead to the same conjugacy-preserving TOA operator for the free-particle case, they are completely different for the interacting case. Only the Weyl-quantized TOA-operator satisfies the conjugacy requirement but is limited only for linear systems. For nonlinear systems, no quantized operator satisfies the required commutator values. 

On another front, one of us has reconsidered the quantum TOA problem by constructing TOA operators independent of canonical quantization. Reference \cite{Galapon2012} introduced the idea of supraquantization where quantum observables are constructed purely from quantum mechanical considerations. In this formulation, the quantum image of the classical arrival time expression is obtained by mapping $\mathcal{T}_0(q,p)$ as a specific TOA operator $\hat{\mathrm{T}}_S$ in the rigged Hilbert space $\Phi^\times \supset L^2(\mathcal{R}) \supset \Phi $. The space $\Phi$ is the fundamental space of infinitely differentiable functions in the real line with compact supports and $\Phi^\times$ is the space of functionals on $\Phi$. The corresponding time operator $\hat{\mathrm{T}}_S$, called the supraquantized TOA operator, has been obtained by directly imposing significant physical properties similar to the quantized ones: (i) hermiticity: $\hat{\mathrm{T}}_S=\hat{\mathrm{T}}_S^\dagger \to \langle q|\hat{\mathrm{T}}_S|q'\rangle^*=	\langle q|\hat{\mathrm{T}}_S|q'\rangle;$
(ii) time-reversal symmetry: $\Theta \hat{\mathrm{T}}_S \Theta ^{-1}=-\hat{\mathrm{T}}_S \to \langle q|\hat{\mathrm{T}}_S|q'\rangle^*=-\langle q|\hat{\mathrm{T}}_S|q'\rangle;$
and the (iii) correspondence between classical and quantum arrival times established by the known Weyl-Wigner transform
\begin{equation}\label{correspondence}
	\mathcal{T}_0(q,p)\,=\,\lim_{\hbar \to 0}\, \mathcal{T}_\hbar(q,p)\,=\,\lim_{\hbar \to 0}\,\int_{-\infty}^{\infty}d\nu \left\langle q+\frac{\nu}{2}\left|\hat{\mathrm{T}}_S\right|q-\frac{\nu}{2}\right\rangle \mathrm{e}^{-i \nu p/\hbar}. 
\end{equation}
It turns out that $\hat{\mathrm{T}}_S$ shares the same functional structure as Eq. (\ref{toaopformq}). The difference, however, lies in its time kernel factor $T_S(q,q')$ where it appears instead as a solution to a specific second order partial differential equation  
\begin{equation}\label{TKEorigintro}
-\frac{\hbar^2}{2\mu} \frac{\partial^2 T_S(q,q')}{\partial q^2}+\frac{\hbar^2}{2\mu} \frac{\partial^2 T_S(q,q')}{\partial q'^2}+ \left[V(q)-V(q')\right]T_S(q,q')=0,
\end{equation}
called the time kernel equation (TKE) and is subject to the boundary conditions $T_S(q,q)=q/2$ and $T_S(q,-q)=0.$  The TKE, along with the specified boundary conditions, admits a unique solution for entire analytic potentials \cite{Sombillo2012}. It is a direct consequence of the conjugacy requirement of the supraquantized time of arrival operator with the system Hamiltonian, i.e., $	[\hat{\mathrm{H}},\hat{\mathrm{T}}_S]=i\hbar \mathbb{1}$. Hence, we have referred to the supraquantized operator $\hat{\mathrm{T}}_S$ as the conjugacy-preserving TOA (CPTOA) operator. On the other hand, we call operators whose time kernel factors do not satisfy the TKE as non-conjugacy-preserving TOA (NCPTOA) operators. In essence, the construction of conjugacy-preserving TOA operators translates to the problem of solving the TKE \eqref{TKEorigintro} for a given interaction potential.

Now, the general solution of the TKE admits the expansion $T_S(q,q')=T_W(q,q')+\sum_{n=1}^{\infty}T_n(q,q')$, where the leading term $T_W(q,q')$ is identified as the Weyl-quantized TKF (\ref{tweyl}). The succeeding terms $T_n(q,q')$ for $n \ge 1$ are interpreted as quantum corrections to $T_W(q,q')$ whose inverse Weyl-Wigner transforms ultimately vanish in the classical limit $\hbar \to 0$. As a consequence, the supraquantized TOA operator formally assumes the expansion
\begin{equation}\label{supraexpand}
\hat{\mathrm{T}}_S=\hat{\mathrm{T}}_W+\sum_{n=1}^{\infty}\hat{\mathrm{T}}_n,
\end{equation} 
where the leading operator $\hat{\mathrm{T}}_W$ is recognized as the Weyl-quantized TOA operator \cite{Pablico2023}. The succeeding operators $\hat{\mathrm{T}}_n$ are required to satisfy the desired commutation relation.

For all intents and purposes, Eq. \eqref{supraexpand} represents the complete conjugacy-preserving TOA operator. Its functional form is particularly advantageous when performing perturbative analysis on specific quantum arrival time problems as one can isolate each perturbative order corresponding to every quantum corrective term. For such cases, one can study the leading dominant term, the Weyl-quantized operator, for its general behavior and then further accommodate additional terms sufficient for the intended accuracy and purpose \cite{Sombillo2012}. For example, when studying the quantum traversal and tunneling times across a potential barrier, each term in the expansion may represent time measurements in different regimes. In particular, if the leading term corresponds to a measurement in the attosecond regime, which is a very short time interval measured in attosecond ionization experiments, the next term might correspond to even shorter time intervals, like zeptoseconds or even subzeptosecond domains. In this case, a vanishing traversal time in the attosecond timescale does not necessarily imply vanishing time measurements in all smaller time scales. 

Of course, there are cases when the practical utility of the prescribed expanded solution is limited. This is because each time kernel factors are expressed as recurrence relations in integral form (see Eq. (\ref{tnuvintro})). Although theoretically solvable iteratively, the difficulty of constructing $T_n(q,q')$'s escalates remarkably for higher order $n$. Thus, the expanded iterative solution of the time kernel equation may pose substantial difficulties for both analytical and numerical implementations. 

One specific problem that calls for an exact-solution of the TKE is the investigation of the dynamical properties of the complete CPTOA operator. While Eq. (\ref{supraexpand}) might already suggest the desired dynamical characteristics, primarily through its leading term, it may not be sufficient as the additional operators $\hat{\mathrm{T}}_n$ may actually introduce corrections that deviate instead the leading term from the ideal unitary behavior of time operators at the arrival point. Furthermore, an exact-closed form solution is needed to fully grasp the distinction between operators constructed by quantization and those that are constructed by supraquantization in terms of the time evolution of their eigenfunctions. This analysis should help us determine if the canonical commutation relation between time and energy influences the observed dynamics of time operators.

Thus, in this paper, we revisit the method of supraquantization with three specific objectives in mind. Our first objective is to find an exact closed-form solution of the time kernel equation with the goal of expressing the conjugacy-preserving TOA operator as a single integral operator, rather than an infinite series of operators as defined in Eq. (\ref{supraexpand}). This is achieved by considering a specific class of potentials satisfying a particular separability condition. Our second objective is to delve into the essential properties of the CPTOA operator and formally demonstrate its unitary arrival property by coarse graining methods and spatial confinement. Our third and final objective is to compare the dynamics obtained from the conjugacy-preserving TOA operator with that of a non-conjugacy-preserving TOA operator obtained by Weyl quantization of the classical arrival time. We check if there is a discernible difference between the two operators with respect to their dynamics. 

The analysis of TOA operators plays a pivotal role in various practical applications across quantum mechanics and related fields. They provide insights into the temporal aspects of quantum phenomena and help resolve paradoxes related to time measurements, such as quantum tunneling \cite{Galapon2012}, the Hartman effect \cite{Sombillo2012}, quantum backflow \cite{backflow}, and the quantum Doppler effect \cite{doppler}, among others. Their practical applications may also extend to: (i) quantum computing, where precise timing of quantum gates and operations is essential for designing efficient algorithms; (ii) quantum optics, where studying photon arrival times aids in developing new optical devices and technologies; (iii) semiconductor physics, where understanding quantum tunneling and traversal effects is crucial for devices like tunnel diodes and quantum dots; and (iv) chemical physics, where analyzing the timing and dynamics of molecular interactions and reactions is vital for designing efficient chemical processes and developing new materials.

The paper is structured as follows. In Section (\ref{sec:reviewsupraexpand}), we offer a brief overview of the expanded iterative solution of the time kernel equation. Section (\ref{sec:closedformsol}) covers our first objective, presenting the derivation of an exact closed-form solution for the time kernel equation in integral form. We will also establish the necessary and sufficient conditions for separability, accompanied by a listing of separable interaction potentials. Section (\ref{sec:dynamics}) addresses our remaining objectives, where we explore the time evolution of CPTOA operator eigenfunctions and compare them with the Weyl-quantized operator for different interaction potentials. Finally, in Section (\ref{sec:conclusion}), we present our concluding remarks.

\section{Expanded iterative solution of the TKE} \label{sec:reviewsupraexpand}

Let us review the prescribed expanded iterative solution of the TKE \cite{Pablico2023}. From here on, we work on the canonical coordinates defined by the transformation $(u=q+q',\, v=q-q')$. In these new coordinates, the TKE assumes the canonical form
\begin{equation}\label{tkeuv}
	-\frac{2\hbar^2}{\mu} \frac{\partial^2 T_S(u,v)}{\partial u \partial v}+\left(V\left(\frac{u+v}{2}\right)-V\left(\frac{u-v}{2}\right)\right)T_S(u,v)=0,
\end{equation}
subject to the following boundary conditions
\begin{equation}\label{boundary}
	T_S(u,0)=\frac{u}{4} ;      \qquad  T_S(0,v)=0.
\end{equation}
Integrating both sides of Eq. (\ref{tkeuv}) with respect to $u$ and $v$ gives the integral form of the TKE
\begin{equation}\label{tkeint}
T_S(u,v)=\frac{u}{4}+\frac{\mu}{2\hbar^2}\int_{0}^{u} du'\,\int_{0}^{v} dv'\,\left(V\left(\frac{u+v}{2}\right)-V\left(\frac{u-v}{2}\right)\right)\,T_S(u',v').
\end{equation}
For practical purposes, it is this integral form of the TKE that we will solve in this paper.

We now consider the solution of the TKE for arbitrary analytic potentials of the form $V(q)=\sum_{s=1}^{\infty} a_s q^s$,
where $a_s$ are some expansion coefficients. We refer to linear (nonlinear) systems as those whose interaction potential $V(q)$ yields a linear (nonlinear) classical equation of motion in $q$. For linear systems, $a_s=0$ for $s\ge3$ while the coefficient $a_3$ at least needs to be nonvanishing for nonlinear systems. The corresponding TKE appears as
\begin{equation} \label{tke2}
-2 \frac{\hbar^2}{\mu} \frac{\partial^2 T_S(u,v)}{\partial u \partial v} + \sum_{s=1}^{\infty} \frac{a_s}{2^{s-1}} \sum_{k=0}^{[s]} \binom{s}{2k+1}\, u^{s-2k-1} \,  v^{2k+1} \, T_S(u,v)=0,
\end{equation}
where $[s]=(s-1)/2$ for odd $s$ and  $[s]=s/2-1$ for even $s$. A possible solution of Eq. (\ref{tke2}) can be obtained by assuming an analytic solution in powers of $u$ and $v$ defined by $T_S(u,v) = \sum_{m,n=0}^{\infty} \alpha_{m,n} u^m v^n
$, for some unknown coefficients $\alpha_{m,n}$ satisfying the condition $\alpha_{m,n}=0$ for $m,n<0$. The original boundary conditions of the TKE imply the initial conditions $\alpha_{m,0}=\delta_{m,1}/4$, and $\alpha_{0,n}=0$ for all $m$ and $n$.

Substituting back our assumed solution into Eq. (\ref{tke2}) and performing particular series rearrangements, one finds the expansion
\begin{equation}\label{tkesol3compact}
T_S(u,v) =T_{0}(u,v) +  \sum_{n=1}^{\infty}T_n(u,v),
\end{equation}
where the leading and succeeding terms are identified as 
\begin{equation}\label{tkesolT0gen}
T_{0}(u,v) = \sum_{m=0}^{\infty} \sum_{j=0}^{\infty} u^{m} v^{2j} 	\left(\frac{\mu}{2 \hbar^2}\right)^j \alpha_{m,j}^{(0)},
\end{equation}
\begin{equation}\label{tkesolTn}
T_n(u,v)=\sum_{m=0}^{\infty}\sum_{j=0}^{\infty}u^m v^{2j+2n+2}\left(\frac{\mu}{2 \hbar^2}\right)^{j+1}\alpha^{(n)}_{m,j+n+1}.
\end{equation}
The new coefficients $\alpha_{m,j}^{(s)}$ satisfy the following recurrence relation
\begin{equation}\label{alphasmj}
\alpha_{m,j}^{(s)}=\frac{1}{m \cdot 2j} \sum_{r=0}^{s}\sum_{l=2r+1}^{m+2r-1}\frac{a_l}{2^{l-1}}\,\binom{l}{2r+1} \, \alpha_{m-l+2r,j-r-1}^{(s-r)},
\end{equation}
for all $0 \le s \le (j-1)$ with the boundary conditions $\alpha_{m,j}^{(0)}=\delta_{m,1}/4$ and $\alpha_{m,j}^{(s)}=0$ for $m,j \le0$. 

Equation (\ref{tkesol3compact}) has been obtained such that its full Weyl-Wigner transform appears in powers of $\hbar$, that is,
\begin{equation}\label{wignertrans}
\begin{split}
\mathcal{T}_\hbar(q,p)&=\int_{-\infty}^{\infty}d\nu \,\left\langle q+\frac{\nu}{2}\left|\hat{\mathrm{T}}\right|q-\frac{\nu}{2}\right\rangle \mathrm{e}^{-i \nu p/\hbar}\\
&=\sum_{n=0}^{\infty}\frac{\mu}{i\hbar}\int_{-\infty}^{\infty}d\nu\,\, T_n\left(q+\frac{\nu}{2},q-\frac{\nu}{2}\right)\,\mathrm{sgn}(\nu)\,\mathrm{e}^{-ip\nu\hbar}\\
&=\tilde{\tau}_0(q,p)+\hbar^2 \, \tilde{\tau}_1(q,p) + \hbar^4 \,  \tilde{\tau}_2(q,p)+... .
\end{split}
\end{equation}
The leading term of Eq. (\ref{wignertrans}) simplifies to the classical TOA (\ref{classicaltoa}) while the terms $\hbar^{2n} \, \tilde{\tau}_n(q,p)$ for $n \ge 1$, where the factors $\tilde{\tau}_n(q,p)'$s are independent of $\hbar$, ultimately vanish in the classical limit. One can then interpret the $\hbar-$dependent terms as the quantum corrections to the classical arrival time \cite{Galapon2004,Pablico2023}. 

Following Eq. (\ref{correspondence}), we can then identify the kernel  $T_0(u,v)$ as the Weyl quantization of the classical arrival time and the succeeding terms $T_n(u,v)$ (\ref{tkesol3compact}) as the Weyl map of the quantum corrections $\hbar^{2n} \, \tilde{\tau}_n(q,p)$. This suggests that the terms $T_n(u,v)$ could be interpreted as the quantum corrections to the Weyl-quantized time kernel factor $T_0(u,v)$. Understanding this expanded solution aids in exploring the link between classical and quantum observables through the relation $\mathcal{T}_0(q,p)=\lim_{\hbar \to 0}\, \mathcal{T}_\hbar(q,p)$ \cite{Galapon2001}. In Ref. \cite{PablicoMoyal2023}, we have demonstrated that these quantum corrections also arise from the Moyal bracket of an arrival time observable with the system Hamiltonian. This Moyal bracket is a deformation of the Poisson bracket relation and is fundamental in deformation quantization and quantum phase space \cite{Zachos2005}.

Now, solving Eq. (\ref{tkesolTn}) in terms of generating functions and employing the method of successive approximations, one finds the functional form of the kernel factors $T_n(u,v)$ given by
\begin{equation}\label{tnuvintro}
\begin{split}
T_n(u,v)=&\left(\frac{\mu}{2\hbar^2}\right)\sum_{r=1}^{n} \frac{1}{(2r+1)!}\frac{1}{2^{2r}}\int_{0}^{u} ds \, V^{(2r+1)}\left(\frac{s}{2}\right) \int_{0}^{v} dw \, w^{2r+1} \, T_{n-r}(s,w)\,G(s,w),
\end{split}
\end{equation}
for all $n \ge 1$, with the factor $G(s,w)$ defined as
\begin{equation}
G(s,w)={}_0F_1 \left(;1;\left(\frac{\mu}{2\hbar^2}\right)(v^2-w^2)\left[V \left(\frac{u}{2}\right)-V \left(\frac{s}{2}\right)\right]\right).
\end{equation}
The convergence of the above result is guaranteed by the continuity of the interaction potential $V(q)$ and the absolute convergence of the hypergeometric function ${}_pF_q(a;b;z)$ for $p<q$. The Weyl-quantized time kernel factor $T_0(u,v)=T_W(u,v)$ (\ref{tweyl}) appears as the initial condition of Eq. (\ref{tnuvintro}).

We call the solution defined in Eq. (\ref{tkesol3compact}) as the expanded iterative solution of the TKE since it appears as an infinite series, with the terms $T_n(u,v)$ appearing as some specific integrals of the previous iterates $T_{n-1}(u,v)$. For example, the first two TKF corrections appear explcitly as
\begin{equation}\label{t10uvfinb}
\begin{split}
T_1(u,v)&=\left(\frac{\mu}{48\hbar^2}\right)\int_{0}^{u} ds \, V'''\left(\frac{s}{2}\right) \int_{0}^{v} dw \, w^3 \, T_0(s,w)\,\,G(s,w),
	\end{split}
\end{equation}
\begin{equation}\label{t20uv}
\begin{split}
T_2(u,v)&=\frac{1}{4 \cdot 3!}\left(\frac{\mu}{2\hbar^2}\right)\int_{0}^{u} ds \, V^{(3)}\left(\frac{s}{2}\right) \int_{0}^{v} dw \, w^3 \, T_1(s,w) \, G(s,w)\\
&+\frac{1}{16 \cdot 5!}\left(\frac{\mu}{2\hbar^2}\right)\int_{0}^{u} ds \, V^{(5)}\left(\frac{s}{2}\right) \int_{0}^{v} dw \, w^5 \, T_0(s,w) \, G(s,w).
\end{split}
\end{equation}
Note that each $T_n(u,v)$ term can be derived from Eq. (\ref{tnuvintro}) through iterative methods. Essentially, this means that one could construct the entire supraquantized time kernel factor for a given interaction potential, thereby constructing a complete conjugacy-preserving TOA operator.

\section{Exact closed-form solution of the TKE for separable potentials}\label{sec:closedformsol}

\subsection{Construction of the solution}

Performing analytical and numerical analysis on a time operator that involves an infinite number of integral operators is clearly not an easy task. Hence, let us now accomplish our first objective and derive an exact closed-form solution of the TKE in its integral form. We do so by considering a specific class of potentials which we refer as separable potentials. 

\begin{definition} [Separable potential] \label{Def:SeparablePot}
    A potential $V(q)$ is separable if satisfies the condition 
    \begin{equation}\label{separability}
        V\left(\frac{u+v}{2}\right)-V\left(\frac{u-v}{2}\right) = F(u)\,G(v),
    \end{equation}
    where $F(u)$ and $G(v)$ are univariate functions of $u$ and $v$, respectively, and referred to as the divisors of $V(q)$.
\end{definition}

\noindent An advantage of the above condition is that it will help us decouple later the double integral along $u$ and $v$ in Eq. (\ref{tkeint}). Assuming our interaction potential satisfies the separability condition defined in Eq. (\ref{separability}), the TKE can now be cast into the following integral equation
\begin{equation}\label{tkeintt}
T_S(u,v)=\frac{u}{4}+\frac{\mu}{2\hbar^2}\int_{0}^{u} du'\,\,F(u')\int_{0}^{v} dv'\,G(v')\,T_S(u',v').
\end{equation}
The necessary and sufficient conditions for separability, along with examples of separable interaction potentials, will be discussed in Section (\ref{sec:conditionssep}). Note that Eq. (\ref{tkeintt}) does not apply to all analytic potentials. Nevertheless, considering separable potentials is sufficient for our purposes, as it lets us examine the role of conjugacy in the dynamics of TOA operators.

We solve for the desired kernel factor $T_S(u,v)$ using the method of successive approximations \cite{Pablico2023}. We start by choosing the leading term as our zeroth-order approximation, that is, $T_S^{(0)}(u,v)=u/4$, so that the boundary conditions $T_S(u,0)=u/4$ and $T_S(0,v)=0$ readily emerge when either $u=0$ or $v=0$, respectively. The $n$th order approximation of the kernel factor $T_S(u,v)$ is obtained from the recurrence relation
\begin{equation}\label{tnuviterate}
T^{(n)}_S(u,v)=\frac{u}{4}+\frac{\mu}{2\hbar^2}\int_{0}^{u} du'\,\int_{0}^{v} dv'\,F(u')\,G(v')\,T^{(n-1)}_S(u',v'),
\end{equation}
valid for $n \ge 1$. Notice that the comparison between Eqs. (\ref{tkeintt}) and (\ref{tnuviterate}) suggests the relation $T_S(u,v)=\lim_{n\to\infty}T^{(n)}_S(u,v).$ This means that the solution of the TKE follows immediately from the solution of the recurrence relation for $T^{(n)}_S(u,v)$ (\ref{tnuviterate}). 

By iteration, the first three approximations of Eq. (\ref{tnuviterate}) assume the following specific forms
\begin{equation}\label{t1suv}
T^{(1)}_S(u,v)=\frac{u}{4}+\frac{\mu}{2\hbar^2}\,\int_{0}^{v} dv_1\,G(v_1)\,\int_{0}^{u} du_1F(u_1)\,\frac{u_1}{4},
\end{equation}
\begin{equation}\label{t2suv}
\begin{aligned}
T^{(2)}_S(u,v)=&\frac{u}{4}+\frac{\mu}{2\hbar^2}\,\int_{0}^{v} dv_1\,G(v_1)\,\int_{0}^{u} du_1F(u_1)\,\frac{u_1}{4}\\
&+\left(\frac{\mu}{2\hbar^2}\right)^2\,\int_{0}^{v} dv_1\,G(v_1)\,\int_{0}^{v_1} dv_2\,G(v_2)\,\int_{0}^{u} du_1F(u_1)\,\int_{0}^{u_1} du_2F(u_2)\,\frac{u_2}{4},
\end{aligned}
\end{equation}
\begin{equation}\label{t3suv}
 \begin{aligned}
&T^{(3)}_S(u,v)=\frac{u}{4}+\frac{\mu}{2\hbar^2}\,\int_{0}^{v} dv_1\,G(v_1)\,\int_{0}^{u} du_1F(u_1)\,\frac{u_1}{4}\\
&+\left(\frac{\mu}{2\hbar^2}\right)^2\,\int_{0}^{v} dv_1\,G(v_1)\,\int_{0}^{v_1} dv_2\,G(v_2)\,\int_{0}^{u} du_1F(u_1)\,\int_{0}^{u_1} du_2F(u_2)\,\frac{u_2}{4}\\
&+\left(\frac{\mu}{2\hbar^2}\right)^3\int_{0}^{v} dv_1\,G(v_1)\int_{0}^{v_1} dv_2\,G(v_2)\int_{0}^{v_2} dv_3\,G(v_3)\int_{0}^{u} du_1F(u_1)\int_{0}^{u_1} du_2F(u_2)\int_{0}^{u_2} du_3F(u_3)\frac{u_3}{4}.
\end{aligned}   
\end{equation}
We observe that the leading term is always the zeroth-order approximation $u/4$. On the other hand, the succeeding $n$th term involves $(n-1)$ iterated integrals along $u$ and $v$. Generalizing the above results for arbitrary $n\ge1$, we get
\begin{equation}\label{tsnuv}
\begin{aligned}
T^{(n)}_S(u,v)& =\frac{u}{4}+\frac{\mu}{2\hbar^2}\,\int_{0}^{v} dv_1\,G(v_1)\,\int_{0}^{u} du_1F(u_1)\,\frac{u_1}{4}\\
&+\left(\frac{\mu}{2\hbar^2}\right)^2\,\int_{0}^{v} dv_1\,G(v_1)\,\int_{0}^{v_1} dv_2\,G(v_2)\,\int_{0}^{u} du_1F(u_1)\,\int_{0}^{u_1} du_2F(u_2)\,\frac{u_2}{4}+\,...\\
&+\left(\frac{\mu}{2\hbar^2}\right)^n\int_{0}^{v} dv_1\,G(v_1)\,\,...\,\int_{0}^{v_2} dv_{n-1}\,G(v_n)\int_{0}^{u} du_1F(u_1)\,\,...\,\int_{0}^{u_{n-1}} du_nF(u_n)\frac{u_n}{4}.
\end{aligned}
\end{equation}
Our expression for $T^{(n)}_S(u,v)$ is formally proven by mathematical induction. Assuming that the above equation holds true for some $n=k \ge 0$, the next iterate, $n=k+1$, also holds true.

While Eq. (\ref{tsnuv}) looks complicated at first, it can be easily evaluated by exploiting the following integral identity \cite{Zwillinger1992}
\begin{equation}\label{identityiterate}
\int_a^xdx_1 \, g(x_1) \int_a^{x_1}dx_2 \, g(x_2)\,\,...\, \int_a^{x_{n-1}}dx_n g(x_n)\, \int_a^{x_n}dt\,f(t) = \frac{1}{n!}\,\int_a^{x}dt\,f(t) \left[\int_t^x\,dy\,g(y)\right]^n.
\end{equation}
Applying the above result to Eq. (\ref{tsnuv}), we arrive at 
\begin{equation} \label{tsfullexpand}
\begin{aligned}
T^{(n)}_S(u,v) = \frac{u}{4} + \left(\frac{\mu}{2\hbar^2}\right) &\sum_{r=1}^n \frac{1}{(r-1)!(r-1)!}\left(\frac{\mu}{2\hbar^2}\right)^{r-1}\int_{0}^{v} dv'\,G(v')\,\left[\int_{v'}^{v} dv''\,G(v'')\right]^{n-1}\\
\times &\int_{0}^{u} du'\,F(u')\frac{u'}{4}\,\left[\int_{u'}^{u} du''\,F(u'')\right]^{n-1}.
\end{aligned}
\end{equation}

Now, taking the limit $n\to \infty$ and shifting index from $r \to r-1$, we finally get the desired exact closed-form solution of the TKE and is formally given by
\begin{equation}\label{tsfull}
T_S(u,v)=\frac{u}{4}+\frac{\mu}{2\hbar^2}\int_{0}^{v} dv'\,G(v')\, \int_{0}^{u} du'\,F(u')\frac{u'}{4} \,{}_0F_1 \left(;1;\left(\frac{\mu}{2\hbar^2}\right)\,\Tilde{G}(v,v')\,\Tilde{F}(u,u')\right),
\end{equation}
where the factors $\Tilde{G}(v,v')$ and $\Tilde{F}(u,u')$ are defined by
\begin{equation} \label{gvvp}
\Tilde{G}(v,v')=\int_{v'}^{v}dv''\,G(v'')\,, \quad \text{and} \quad \Tilde{F}(u,u')=\int_{u'}^{u} du''\,F(u'').
\end{equation}
In terms of the original interaction potential, Eq. (\ref{tsfull}) expands to
\begin{equation}\label{tsfull2}
\begin{aligned}
T_S(u,v)=\frac{u}{4}+\frac{\mu}{2\hbar^2}\int_{0}^{v} dv'&\,\, \int_{0}^{u} du'\,\frac{u'}{4} \,\left[ V\left(\frac{u'+v'}{2}\right)-V\left(\frac{u'-v'}{2}\right)\right]\\
&\times{}_0F_1 \left(;1;\left(\frac{\mu}{2\hbar^2}\right)\,\int_{v'}^{v}dv''\int_{u'}^{u} du''\left[ V\left(\frac{u''+v''}{2}\right)-V\left(\frac{u''-v''}{2}\right)\right]\right).
\end{aligned}
\end{equation}
In essence, the above equation fulfills our first objective. 

The original boundary conditions $T_S(u,0)=u/4$ and  $T_S(0,v)=0$ immediately follow from our full solution. In addition, the real-valuedness of the divisors $F(u)$ and $G(v)$ guarantee the real-valuedness of the kernel factor $T_S(u,v)$. Hence, the full supraquantized operator is indeed Hermitian. On the other hand, the full kernel factor is symmetric in the original $(q=(u+v)/2,q'=(u-v)/2)$ coordinates, that is, $T_S(q,q')=T_S(q',q)$. Hence, the corresponding supraquantized operator also satisfies time-reversal symmetry. Finally, one can show that $T_S(u,v)$ leads to the correct classical arrival time expression in the classical limit in accordance to Eq. (\ref{correspondence}). Hence, all three key properties of a TOA observable are satisfied. As a final check, it can also be shown that the derived solution correctly satisfies the TKE, as presented in Appendix (\ref{app:finalchecktke}). 

Interestingly, the first term of Eq. (\ref{tsfullexpand}) represents the time kernel factor of the free TOA operator, while the remaining terms depends on the interaction potential $V(q)$. This suggests that the supraquantized TOA operator $\hat{\mathrm{T}}_S$ can be expressed as the sum of the free TOA operator $\hat{\mathrm{T}}_F$ and some specific potential-dependent operator, call $\hat{\mathrm{T}}_V$, i.e.,
 \begin{equation}\label{supraexpandfin}
\hat{\mathrm{T}}_S=\hat{\mathrm{T}}_F+\hat{\mathrm{T}}_V,
\end{equation} 
where the kernel factor of the operator $\hat{\mathrm{T}}_V$ is simply the second term of Eq. (\ref{tsfull2}). As the potential approaches to zero, the second term vanishes. For sufficiently weak interaction potentials, the operator $\hat{\mathrm{T}}_V$ can be interpreted as a perturbation to the free TOA operator $\hat{\mathrm{T}}_F$. For such cases, one can perform perturbative analysis similar to the standard perturbation theory of quantum mechanics.

It is important to note that the exact closed-form solution for the TKE we have derived here is applicable only to separable potentials. On the other hand, the expanded iterative solution expressed in Eq. (\ref{tkesol3compact}) holds for any analytic potentials. Nevertheless, both solutions coincide for separable potentials due to the uniqueness of the TKE's solution. It remains an open problem whether a more manageable expression for the CPTOA operator, similar to Eq. (32), can be obtained for arbitrary analytic potentials.

Clearly, it is Eq. (\ref{tsfull2}) (Eq. (\ref{supraexpandfin})) that offers a practical advantage: it involves only two terms which stands in contrast to Eq. (\ref{tkesol3compact}) (Eq. \eqref{supraexpand}) that involves an infinite sum of different time kernel factors (operators). Nevertheless, both solutions are equally important for thoroughly examining the properties of conjugacy-preserving TOA operators and delving into their applications to specific arrival time problems. We will highlight this further in Section (\ref{sec:dynamics}).

\subsection{Necessary and sufficient conditions for separability}\label{sec:conditionssep}

Recall that our ability to derive a closed-form solution for the TKE relies on the assumption that there exists interaction potentials that meet our separability condition, expressed as $V(\left(u+v\right)/2)-V(\left(u-v\right)/2) = F(u)\,G(v)$. In this section, we establish the necessary and sufficient conditions for separability and identify some examples of separable interaction potentials.

\begin{theorem} \label{thm:FormOfSeparableV}
Let $V(q)$ be analytic at the origin. A necessary and sufficient condition for $V(q)$ to be separable is that it admits the functional form
\begin{equation} \label{FormOfSeparableV}
V\left(\frac{u+v}{2}\right)=\frac{1}{2} f(u) g(v) + h(u,v),
\end{equation}
where $g(v)$ and $h(u,v)$ satisfy the following symmetry conditions
\begin{equation}\label{symcondition}
g(-v) = -g(v), \quad \quad  h(u,-v) = h(u,v). 
\end{equation}
The divisors of $V(q)$ are precisely $f(u)$ and $g(v)$.

\end{theorem}



\begin{proof}
\textbf{Necessity.} Since $V(q)$ is analytic at the origin, it assumes the expansion $V(q)=\sum_{n=0}^\infty a_n \, q^n$ with expansion coefficients $a_n$. Using the binomial series expansion $(x+y)^l=\sum_{m=0}^{l} \, \binom{l}{m}\,x^{l-m}\,y^m$, one arrives at
\begin{equation}
V\left(\frac{u+v}{2}\right)= \sum_{n=0}^\infty \sum_{l=0}^n\,\frac{a_n}{2^n}\binom{n}{l}v^{n-l}u^l.
\end{equation}
Reordering the summations using the identity \cite{Srivastava1984}
\begin{equation}\label{identitysum}
\sum_{n=0}^{\infty}\sum_{k=0}^{n}B(k,n)=\sum_{n=0}^{\infty}\sum_{k=0}^{\infty}B(k,n+k), 
\end{equation}
and separating the even and odd terms in $n$, we obtain the following infinite series
\begin{equation}\label{vupv}
V\left(\frac{u+v}{2}\right)= \sum_{l=0}^\infty \sum_{n=0}^\infty\, \binom{2n+l}{l}\,\frac{a_{2n+l}}{2^{2n+l}}\,u^l\,v^{2n}+\sum_{l=0}^\infty \sum_{n=0}^\infty\, \binom{2n+1+l}{l}\,\frac{a_{2n+1+l}}{2^{2n+1+l}}\,u^l\,v^{2n+1}.
\end{equation}
Transforming $v$ to $-v$ and subtracting the resulting equation from Eq. (\ref{vupv}), we obtain 
\begin{equation}\label{vupvm2}
V\left(\frac{u+v}{2}\right)-V\left(\frac{u-v}{2}\right)=2\sum_{l=0}^\infty \sum_{n=0}^\infty\,(-1)^{2n+1}\, \binom{2n+1+l}{l}\,\frac{a_{2n+1+l}}{2^{2n+1+l}}\,u^l\,v^{2n+1}.
\end{equation}

Imposing our separability condition, the above equation can only be separated into a product of univariate functions of $u$ and $v$, say $f(u)$ and $g(v)$, respectively, when the factor $\binom{2n+1+l}{l}\,a_{2n+1+l}$ can be factorized into univariate functions of $n$ and $l$, that is
\begin{equation}\label{factorii}
\binom{2n+1+l}{l}\,a_{2n+1+l}=\alpha(n)\,\beta(l).
\end{equation}
When this condition is satisfied, Eq. (\ref{vupv}) can be rewritten as
\begin{equation}\label{vupv2}
V\left(\frac{u+v}{2}\right)= \sum_{l=0}^\infty \sum_{n=0}^\infty\, \binom{2n+l}{l}\,\frac{a_{2n+l}}{2^{2n+l}}\,u^l\,v^{2n}+\frac{1}{2}\sum_{l=0}^\infty\,\frac{\beta(l)}{2^{2n}}\,u^l \, \sum_{n=0}^\infty\,\frac{\alpha(l)}{2^l}\,v^{2n+1}.
\end{equation}
Note that the first term of Eq. (\ref{vupv2}) is even in $v$. We can define this term as
\begin{equation}\label{huv}
h(u,v)=\sum_{l=0}^\infty \sum_{n=0}^\infty\, \binom{2n+l}{l}\,\frac{a_{2n+l}}{2^{2n+l}}\,u^l\,v^{2n},
\end{equation}
which implies the first symmetry condition $h(u,v)=h(u,-v)$. On the other hand, the absolute convergence of the interaction potential $V(q)$ allows us to express the infinite series in terms of the univariate functions $f(u)$ and $g(v)$ as follows
\begin{equation}\label{fuu}
f(u)=\sum_{l=0}^\infty\,\frac{\beta(l)}{2^{2n}}\,u^l, \quad \quad g(v)=\sum_{n=0}^\infty\,\frac{\alpha(l)}{2^l}\,v^{2n+1}.
\end{equation}
Equation (\ref{fuu}) directly implies the second symmetry condition, $g(v) = -g(-v)$. Substituting Eqs. (\ref{huv}) and (\ref{fuu}) into Eq. (\ref{vupv2}) finally yields Eq. (\ref{FormOfSeparableV}).

\textbf{Sufficiency.} Plugging Eq. \eqref{FormOfSeparableV} into Definition \ref{Def:SeparablePot}, we obtain the following equalities
\begin{align}
V\left(\frac{u+v}{2}\right) - V\left(\frac{u-v}{2}\right) &= \left[\frac{1}{2}f(u) g(v) + h(u,v)\right] -\left[\frac{1}{2}f(u) g(-v) + h(u,-v)\right]\label{vdiff2a} \\
&= \frac{1}{2} f(u) \left[ g(v) - g(-v) \right] + \left[h(u,v) - h(u,-v)\right]\\
&=f(u)\,g(v).\label{vdiff2}
\end{align}
where we have imposed the derived symmetry conditions given by Eq. (\ref{symcondition}) on the second line of Eq. (\ref{vdiff2a}) to arrive at the last equation. Identifying $f(u) = F(u)$ and $g(v) = G(v)$ recovers the original separability condition in Eq. (\ref{separability}).

\end{proof}


We can also consider special class of potential functions with even or odd parity. This is observed by setting $u=v$ and $u=-v$ in Eq. \eqref{separability} and assuming $V(0)=0$. For example, an even $V(q)$ implies the condition $V(u)=V(-u)$. Noting that the divisor $G(u)$ is odd in $u$, one finds the constraint $F(u)=-F(-u)$, implying that the divisor $F(u)$ is also odd in $u$. Likewise, assuming $V(q)$ to be odd in $q$, we also get the condition $F(u)=-F(-u)$, implying $F(u)$ is odd in $u$. These results lead to the following corollary:
\begin{corollary} \label{VFParity}
$F(u)$ is odd when $V(q)$ is even and $F(u)$ is even when $V(q)$ is odd.
\end{corollary}

\begin{theorem}\label{thm:SeparabilityTest}
Let $V(q)$ be analytic at the origin. 
Then $V(q)$ is separable if, for all non-negative integer $m$, it satisfies the following conditions: (1)
\begin{equation}\label{condition2}
V^{(2m+1)}\left(\frac{u}{2}\right)=c_m\,f(u),
\end{equation}
where $V^{(2m+1)}(u/2)$ is the $(2m+1)$-derivative of $V(q)$ evaluated at $u/2$, the $c_m$'s are constants independent of $u$, and $f(u)$ is independent of $m$; and (2) the sum 
\begin{equation}\label{condition3}
 \sum_{m=0}^{\infty} \frac{c_m\,v^{2m+1}}{2^{2m}\,(2m+1)!},
\end{equation}
converges in  some neighborhood of $v=0$. Under these conditions, the divisors of $V(q)$ are given by
\begin{equation}\label{fufu}
F(u)=f(u),
\end{equation}
\begin{equation}\label{gvgv}
G(v)=\sum_{m=0}^{\infty} \frac{c_m\,v^{2m+1}}{2^{2m}\,(2m+1)!} .
\end{equation}
\end{theorem}

\begin{proof}
Assuming the same analytic potential as in Theorem \ref{thm:FormOfSeparableV} and applying the binomial series expansion results in
\begin{equation}
V\left(\frac{u+v}{2}\right)-V\left(\frac{u-v}{2}\right)=\sum_{l=1}^{\infty}\,\frac{a_l}{2^{l-1}}\,\sum_{m=0}^{[l]} \binom{l}{2m+1}u^{l-2m-1}v^{2m+1},
\end{equation}
where the index $[l]$ is defined as $[l]=(l-1)/2$ for odd $l$ and $[l]=l/2-1$ for even $l$. Next, we separate the even and odd parts of the infinite sum along $l$ and then perform a shift in index from $l$ to $l+1$. We arrive at the following expansion
\begin{equation}\label{tkesat5}
\begin{split}
\left(V\left(\frac{u+v}{2}\right)-V\left(\frac{u-v}{2}\right)\right)&=\sum_{l=0}^{\infty}\sum_{m=0}^{l}\,\frac{a_{2l+1}}{2^{2l}}\, \binom{2l+1}{2m+1}u^{2l-2m}v^{2m+1} \\
&+\sum_{l=0}^{\infty}\sum_{m=0}^{l}\,\frac{a_{2l+2}}{2^{2l+1}}\, \binom{2l+2}{2m+1}u^{2l-2m+1}v^{2m+1}.
\end{split}
\end{equation}

Our next goal is to simplify the above expression by closing at least one infinite series. The absolute convergence of our interaction potentials allows us exchange the order of summations along $l$ and $m$ in accordance to Eq. (\ref{identitysum}) which gives 
\begin{equation}\label{tkesat6}
\begin{split}
\left(V\left(\frac{u+v}{2}\right)-V\left(\frac{u-v}{2}\right)\right)=\sum_{m=0}^{\infty}\frac{v^{2m+1}}{2^{2m}(2m+1)!} \left[(2m+1)!\sum_{l=2m+1}^{\infty}a_l\, \binom{l}{2m+1}\left(\frac{u}{2}\right)^{l-(2m+1)}\right],\\
	\end{split}
\end{equation}

Now, we identify that the factor inside the square brackets is simply the $(2m+1)$-derivative of the original potential $V(q)$ evaluated at $q=u/2$. Hence, we arrive at the following result
\begin{equation}\label{vdiffuv}
V\left(\frac{u+v}{2}\right)-V\left(\frac{u-v}{2}\right)=\sum_{m=0}^{\infty} \frac{v^{2m+1}}{2^{2m}\,(2m+1)!}  \, V^{(2m+1)} \left(\frac{u}{2}\right).  
\end{equation}
Clearly when condition \eqref{condition2} holds, the right hand side of Eq. \eqref{vdiffuv} becomes
\begin{equation}\label{separation}
f(u)\sum_{m=0}^{\infty} \frac{c_m v^{2m+1}}{2^{2m} (2m+1)!}.
\end{equation}
If the infinite series in \eqref{vdiffuv} converges in the neighoborhood of $v=0$, then the expression \eqref{separation} is a product of two univariate functions in $u$ and $v$. Finally, imposing our definition for separability (\ref{separability}) on Eq. (\ref{vdiffuv}) gives Eqs. \eqref{fufu} and \eqref{gvgv} following from \eqref{separation}.
\end{proof}

To illustrate these theorems, let us consider the potential $V(q)=A c^{\kappa q} + B c^{-\kappa q}$ for some constants $A$, $B$, $c$ and $\kappa$. It is straightforward to show that
\begin{equation}\label{seprableex1}
V\left(\frac{u+v}{2}\right)-V\left(\frac{u-v}{2}\right)=\big(A c^{\kappa u/2} - B c^{-\kappa u/2}\big)\left(2\sinh\FracArgument{\kappa \, \text{ln} c \, v}{2} \right)=F(u)\,G(v),
\end{equation}
so that $V(q)$ is indeed separable. Let us check if $V(q)$ satisfies our theorems. In $(u,v)$ coordinates, the original interaction potential can be rewritten as
\begin{equation}
\begin{split}
V\left(\frac{u+v}{2}\right)&=\frac{1}{2}\left(A c^{\kappa u/2} - B c^{-\kappa u/2}\right)\left(2\sinh\FracArgument{\kappa \, \text{ln} c \, v}{2} \right) +\left[\left(A c^{\kappa u/2} +  B c^{-\kappa u/2}\right) \cosh\FracArgument{\kappa \, \text{ln} c \, v}{2}\right] \\
&=\frac{1}{2}f(u) g(v) + h(u,v).
\end{split}
\end{equation}
Clearly, $f(u)=F(u)$ and $g(v)=G(v)$. Hence, Theorem (\ref{thm:FormOfSeparableV}) is satisfied. Since the sinh function is odd, $G(v)$ is clearly an odd function of $v$. In addition, taking the $(2m+1)$-derivative of $V(q)$, we arrive at
\begin{equation}
V^{(2m+1)}\left(\frac{u}{2}\right)=\left(\frac{\kappa}{2}\right)^{2m+1} \,\left(A c^{\kappa u/2} - B c^{-\kappa u/2}\right),
\end{equation}
which involves only the original divisor $F(u)$ times the constant $c_m = (\kappa/2)^{2m+1}$. Therefore, Theorem (\ref{thm:SeparabilityTest}) is satisfied. Finally, substitution of the above result into Eqs. (\ref{fufu}) and (\ref{gvgv}), we recover the specific separable condition of the form given by Eq. (\ref{seprableex1}). 

We also list some interaction potentials $V(q)$ in Table (\ref{tablepot}) satisfying our separability condition. One can show that all given potentials satisfy the theorems presented in the current section. 

\begin{table}
\centering
\begin{tabularx}\textwidth{L|M M M}
    \toprule
        
    V(q) & F(u) & G(v) & h(u,v)    \\
        
    \midrule
        
    V_0 q  &  1  &  V_0 v  &  V_0 \frac{u}{2}   \\
        
    V_0 q^2  &   u   &   V_0 v   &   \frac{V_0}{4}(u^2 + v^2)   \\
        
    V_0 \e^{\kappa q}  &  \e^{\kappa u/2}  &  2V_0 \sinh\FracArgument{\kappa v}{2}  &  0   \\
        
    V_0 c^{\kappa q}  &  c^{\kappa u/2}  &  2V_0 \sinh\FracArgument{\kappa \, \text{ln} c \, v}{2}  &  0     \\
        
    A \e^{\kappa q} + B \e^{-\kappa q}  &  A \e^{\kappa u/2} - B \e^{-\kappa u/2}  &  2 \sinh\FracArgument{\kappa v}{2}  &  \left(A \e^{\kappa u/2} +  B \e^{-\kappa u/2}\right) \cosh\FracArgument{\kappa v}{2}   \\
        
    A c^{\kappa q} + B c^{-\kappa q}  &  A c^{\kappa u/2} - B c^{-\kappa u/2}  &  2\sinh\FracArgument{\kappa \, \text{ln} c \, v}{2}  &  \left(A c^{\kappa u/2} + B c^{-\kappa u/2}\right) \cosh\FracArgument{\kappa \, \text{ln} c \, v}{2}   \\
        
    (A \e^{\kappa q} + B \e^{-\kappa q})^2  &  A^2 \e^{\kappa u} - B^2 \e^{-\kappa u}  &  2\sinh(\kappa v)  &  (A^2 \e^{\kappa u} + B^2 \e^{-\kappa u}) \cosh(\kappa v) + 2AB   \\
        
    (A c^{\kappa q} + B c^{-\kappa q})^2  &  A^2 c^{\kappa u} - B^2 c^{-\kappa u}  &  2\sinh(\kappa \, \text{ln} c \, v)  &  (A^2 c^{\kappa u} + B^2 c^{-\kappa u}) \cosh(\kappa \, \text{ln} c \, v) + 2AB    \\
        
    V_0 \sin(\kappa q)  &  \cos\FracArgument{\kappa u}{2}  &  2V_0 \sin\FracArgument{\kappa v}{2}  &  V_0 \sin\FracArgument{\kappa u}{2} \cos\FracArgument{\kappa v}{2}   \\
        
    V_0 \cos(\kappa q)  &  \sin\FracArgument{\kappa u}{2}  &  -2V_0 \sin\FracArgument{\kappa v}{2}  &  V_0 \cos\FracArgument{\kappa u}{2} \cos\FracArgument{\kappa v}{2}   \\ 
        
\bottomrule
\end{tabularx}
\caption{Table of interaction potentials in the configuration space and their corresponding separable forms in $u$-$v$ coordinates. The factors $V_0, \kappa, A, B,$ and $c$ are some arbitrary constants. }\label{tablepot}
\end{table}

\section{Dynamics of the conjugacy-preserving TOA operator}\label{sec:dynamics}

We now proceed to our second objective and establish formally the legitimacy of the supraquantized TOA operator as a genuine TOA operator. In particular, we simulate the time evolution of our operator's eigenfunctions $\varphi_\tau(q)$ and show that its eigenfunctions unitarily arrive at the arrival point at a time equal to its eigenvalue $\tau$. This happens when both the position expectation value equals the arrival point and the position uncertainty is at its minimum, all happening simultaneously at the corresponding eigenvalue time $\tau$ \cite{Galapon2018}.

Our simulation proceeds as follows. First, the TOA eigenvalue problem 
\begin{equation}\label{eigenvalprob}
(\hat{\mathrm{T}}_S\varphi_\tau)(q)=\tau \varphi_\tau(q)=\int_{-\infty}^{\infty}dq' \, \langle q|\hat{\mathrm{T}}_S|q'\rangle \,\varphi_\tau(q'),
\end{equation}
is solved by means of coarse graining. This is done by confining the operator $\hat{\mathrm{T}}_S$ in some closed interval in the real line $[-l,l]$ for some confining length $l>0$, and consequently, approximating the original eigenvalue problem \eqref{eigenvalprob} by replacing the bounds of integrations from $(-\infty,\infty)$ to $(-l,l)$. In essence, we are approximating our original TOA operator $\hat{\mathrm{T}}_S$ as a bounded and self-adjoint Fredholm integral operator. We then discern its behavior for arbitrary large $l$ by successively increasing $l$. The coarse graining can be seen as a deformation of the supraquantized operator in the same configuration space, which is not a significant departure from the original arrival time problem \cite{Galapon2018}. This indicates that observations from the confined case (finite $l$) provide us with similar insights as the unconfined case ($l\to\infty$).

The corresponding eigenvalue problem is then solved by Gauss-Legendre quadrature method. After which, the constructed eigenfunctions are numerically evolved in accordance to the time-dependent Schr\"{o}dinger equation using the operator split method. All parameters and figures that will follow are in atomic units with the parameters $\mu=\hbar=1$. 

\subsection{Unitary arrival property of the supraquantized TOA operator}

Let us consider the potential given by $V(q) = V_0 \cos(kq)$, for some real parameters $V_0$ and $k$, which is typically used in describing optical lattice structures in solid state physics \cite{RevModPhys.78.179}. The supraquantized TOA operator formally assumes the form $\hat{\mathrm{T}}_S=(\mu/i \hbar)\,\int_{-\infty}^{\infty} \,dq' \,T_S(q,q') \,\mathrm{sgn}(q-q')\,\varphi(q')$ with the specific time kernel factor $T_S(q,q')$ given by

\begin{equation}\label{tsupracosine}
T_S(u,v) = \frac{u}{4} - \frac{\mu V_0}{4\hbar^2} \int_{0}^{v} dv' \sin \left( \frac{k v'}{2} \right) \int_{0}^{u} du' u' \sin \left( \frac{k u'}{2} \right) {}_0F_1 \left( ; 1 ; \left(\frac{\mu}{2\hbar^2}\right) \Tilde{G}(v,v') \Tilde{F}(u,u') \right),
\end{equation}
Here, $u=q+q',v=q-q'$ and the factors $\Tilde{G}(v,v')$ and $ \Tilde{F}(u,u')$ are defined as 
\begin{equation}
	\Tilde{G}(v,v') = - \frac{4 V_0}{k} \left[ \cos\left(\frac{k v}{2}\right) - \cos\left(\frac{k v'}{2}\right) \right]; \quad \text{and} \quad \Tilde{F}(u,u') = \frac{2}{k} \left[ \cos\left(\frac{k u}{2}\right) - \cos\left(\frac{k u'}{2}\right) \right].
\end{equation}

The coarse-grained version of the operator $\hat{\mathrm{T}}_S$ has eigenfunctions $\varphi_\tau$ that are square-integrable and eigenvalues $\tau$ that have discrete spectrum. Since the potential is even, the constructed eigenfunctions have definite parities. See also Appendix (\ref{subsec:cptoaparity}) for the discussion of the supraquantized TOA operator under parity transformation. We observe that the eigenfunctions share the same general distinction with that of the free-particle TOA operator and quantized TOA operators for even potentials found in Refs. \cite{Galapon2004a,Galapon2005a,Galapon2018}: they can be either nodal or non-nodal. Nodal eigenfunctions vanish at the arrival point, while non-nodal eigenfunctions do not. 

Figures  (\ref{subfig:nodsupl1t1side}) and (\ref{subfig:nodsupl1t1top}) describe the representative dynamics of the nodal eigenfunctions of the supraquantized operator for the confining length $l=1$ and $V_0=k=1$. As can be seen, the evolved probability density $|\varphi_\tau(q,\tau)|^2$ corresponding to the nodal eigenfunctions have the characteristic dynamical property that two peaks coalesce at the arrival point $q=0$ with their closest approach to each other occurring at the eigenvalue $\tau=0.01$, and then the peaks disperse. In addition, notice that the nodal eigenfunctions are zero at the arrival point at all times. Since the probability density $|\varphi_\tau(q,\tau)|^2$ vanishes at the arrival point, as shown by the blue region indicating the minimum value in Figure (\ref{subfig:nodsupl1t1top}), we interpret the dynamical behavior of the nodal eigenfunctions as the particle arrival with no particle appearance or detection. This intriguing phenomenon coincides with the core result of Ref. \cite{Sombillo2016} which broadens our understanding of quantum arrival. It is possible for a quantum particle to arrive at the arrival point but with no particle appearance. That is, quantum arrival is not synonymous to particle appearance at the arrival point. It may be defined instead as the event when the position uncertainty is at its minimum at the arrival point. Though it may appear counterintuitive, the dynamics of nodal eigenfunctions bear resemblance to the well-known double-slit experiment. If one has no prior knowledge of the particle's path, the viewing screen at the transmission channel exhibits a wave-like nature. In this sense, we can confirm that the particle reached the arrival point, yet particle appearance is no longer guaranteed.

On the other hand, Figures  (\ref{subfig:nondsupl1t1side}) and (\ref{subfig:nondsupl1t1top}) shows the representative dynamics of the non-nodal eigenfunctions of the supraquantized operator for the same length $l=1$. As shown, the corresponding probability density $|\varphi_\tau(q,\tau)|^2$ has the characteristic dynamical property that a single peak forms at the arrival point $q=0$ with its minimum width occurring at the time eigenvalue $\tau=0.01$, after which the peak disperses. Contrary to the nodal eigenfunctions, the probability density is maximum at the origin, as indicated by the yellow region representing the maximum value in Figure (\ref{subfig:nondsupl1t1top}). Thus, the dynamical behavior of the non-nodal eigenfunctions correspond to particle arrival with appearance or detection. If one is only interested with arrival time measurements with particle detection, which is what we observe in real experiments, then the relevant eigenstates that we should consider are the non-nodal eigenfunctions of the supraquantized TOA operator. 

Now, the same unitary arrival property similar to Figure (\ref{fig:nodal_l1_t1_supra}) is also observed for different values of the confining lengths $l$ (See subfigures $(a)$ and $(b)$ in Figures (\ref{fig:weylvssupr_data3_l3_tau2})--(\ref{fig:weylvssupr_data2_10_tau1_noonn})). In our simulations, we have noticed that the probability density becomes more localized for increasing $l$ and fixed $\tau$, $V_0$ and $k$. Hence, we infer that as $l$ approaches infinity, the probability density becomes a function with a singular support, peaking at the arrival point.

These findings fulfill our second objective, demonstrating that the eigenfunctions of the complete supraquantized operators possess the desired unitary arrival property at the origin at their respective eigenvalues within numerical accuracy.


\begin{figure}
\centering
\begin{subfigure}[b]{.32\linewidth}
\includegraphics[width=1\textwidth]{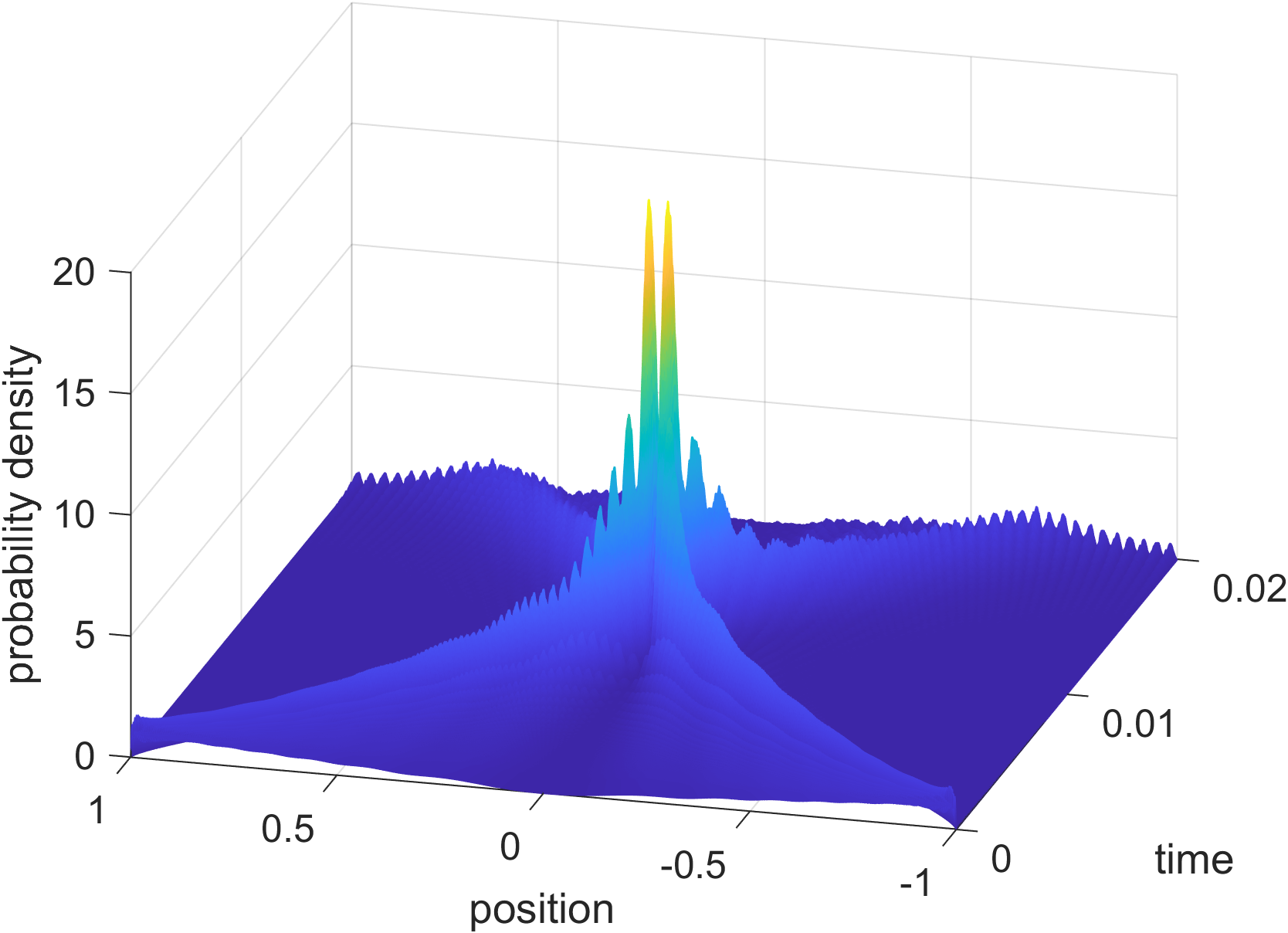}
\caption{Side view: Nodal eigenfunction }\label{subfig:nodsupl1t1side}
\end{subfigure}
\begin{subfigure}[b]{.32\linewidth}
\includegraphics[width=1\textwidth]{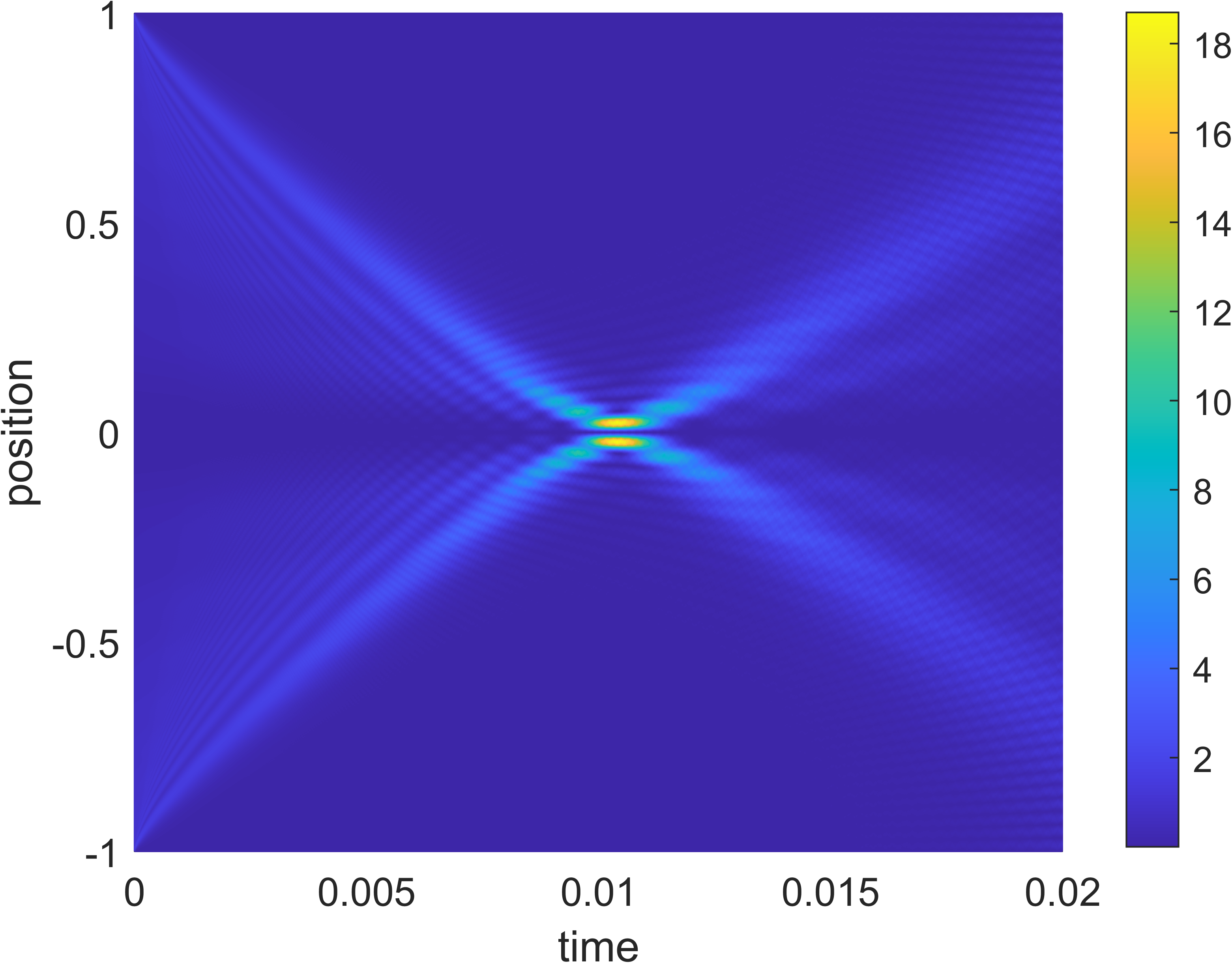}
\caption{Top view: Nodal eigenfunction } \label{subfig:nodsupl1t1top}
\end{subfigure}

\begin{subfigure}[b]{.32\linewidth}
\includegraphics[width=1\textwidth]{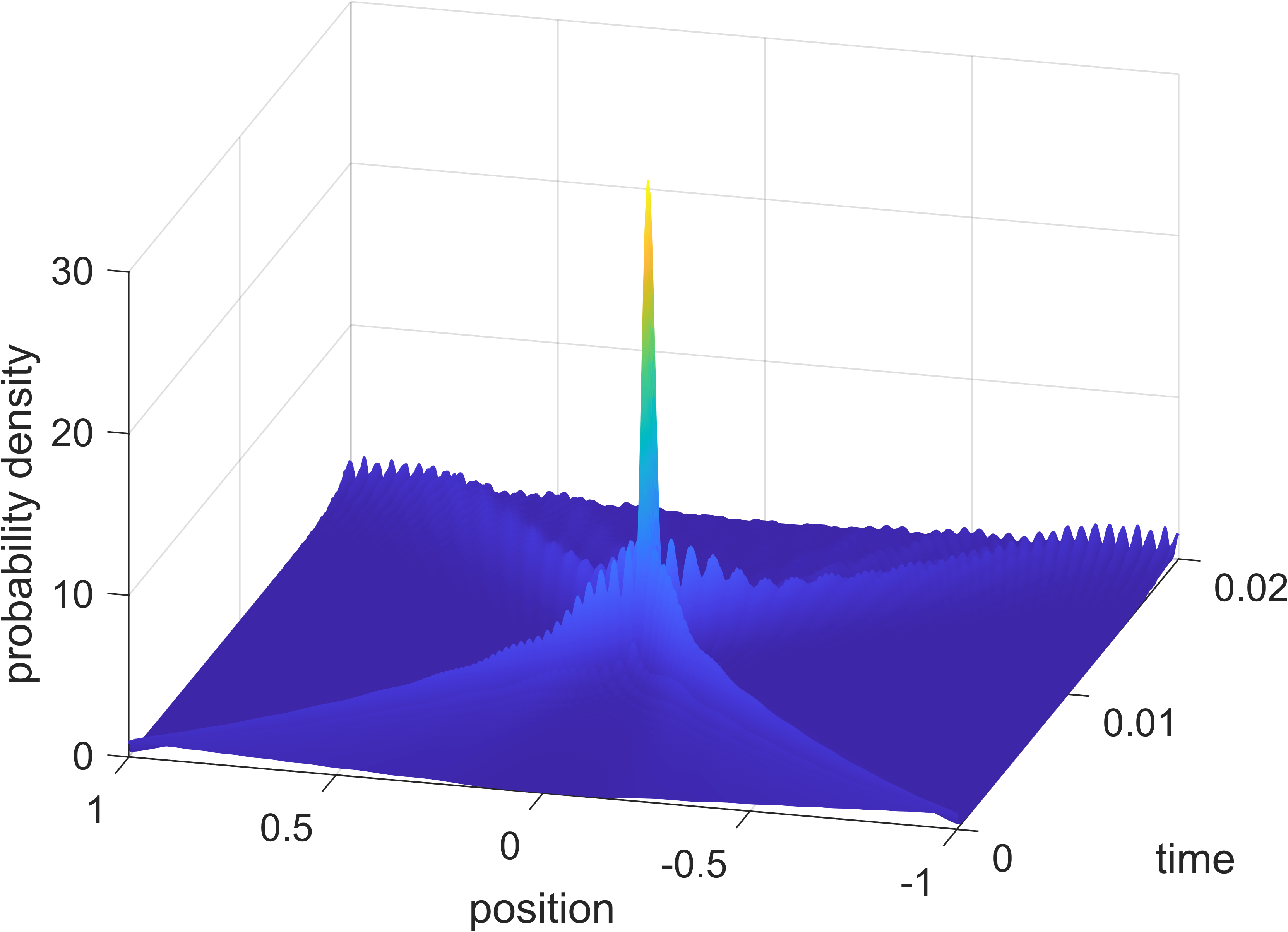}
\caption{Side view: Non-nodal eigenfunction } \label{subfig:nondsupl1t1side}
\end{subfigure}
\begin{subfigure}[b]{.32\linewidth}
\includegraphics[width=1\textwidth]{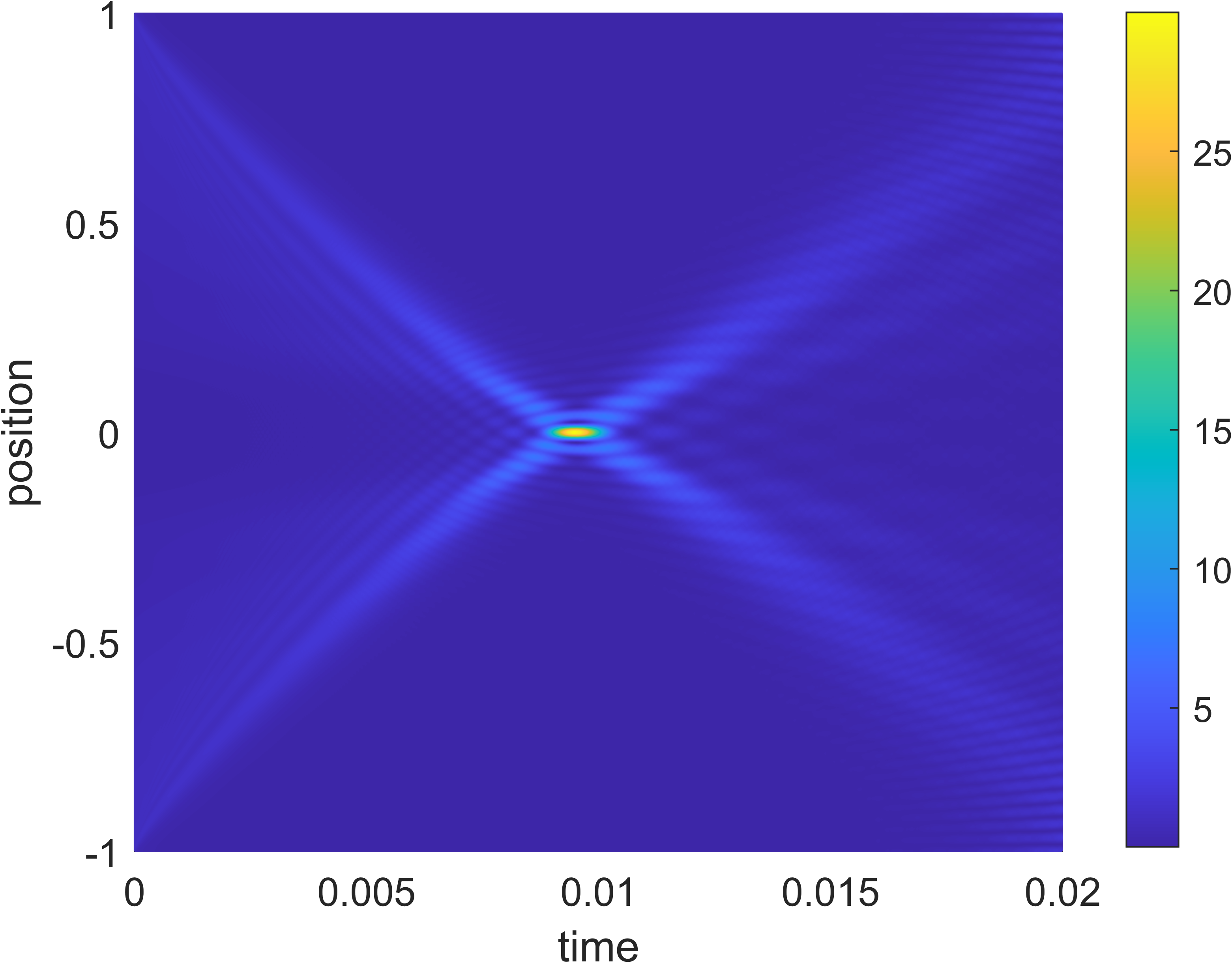}
\caption{Top view: Non-nodal eigenfunction } \label{subfig:nondsupl1t1top}
\end{subfigure}
\caption{The time evolution of the nodal and non-nodal eigenfunctions of the supraquantized TOA operator for the cosine potential with confining length $l=1$ and eigenvalue $\tau=0.01$. (a)--(b) The nodal eigenfunctions show two peaks that coalesce at the arrival point $q=0$ with their closest approach to each other occurring at the eigenvalue $\tau=0.01$. (c)--(d) The non-nodal eigenfunctions unitarily collapse at the arrival point $q=0$ at a time equal to the eigenvalue $\tau$.}\label{fig:nodal_l1_t1_supra}
\end{figure}

\begin{figure}
\centering
\begin{subfigure}[b]{.32\linewidth}
\includegraphics[width=1\textwidth]{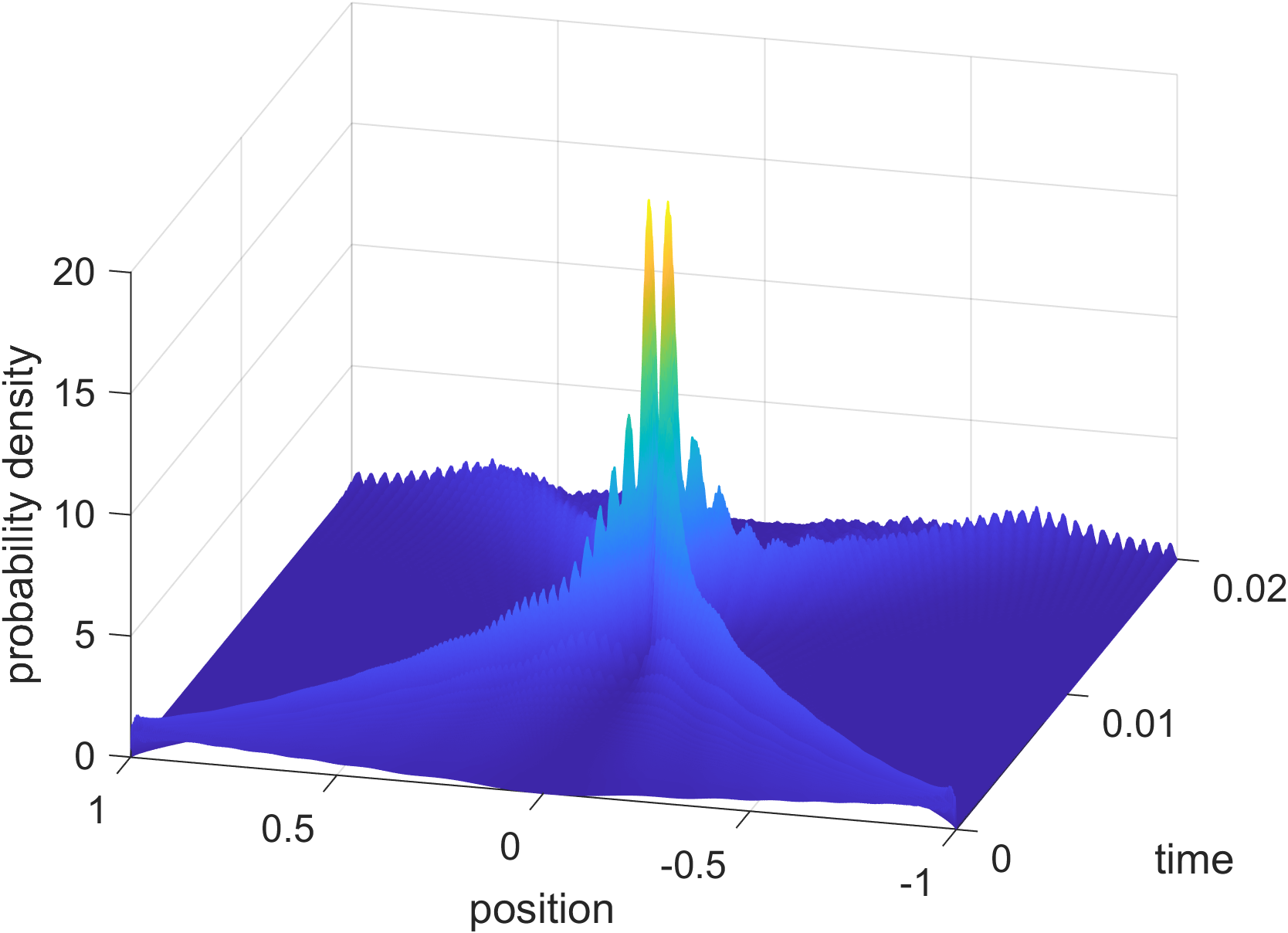}
\caption{Side view: Nodal eigenfunction } \label{subfig:noddweyll1t1side}
\end{subfigure}
\begin{subfigure}[b]{.32\linewidth}
\includegraphics[width=1\textwidth]{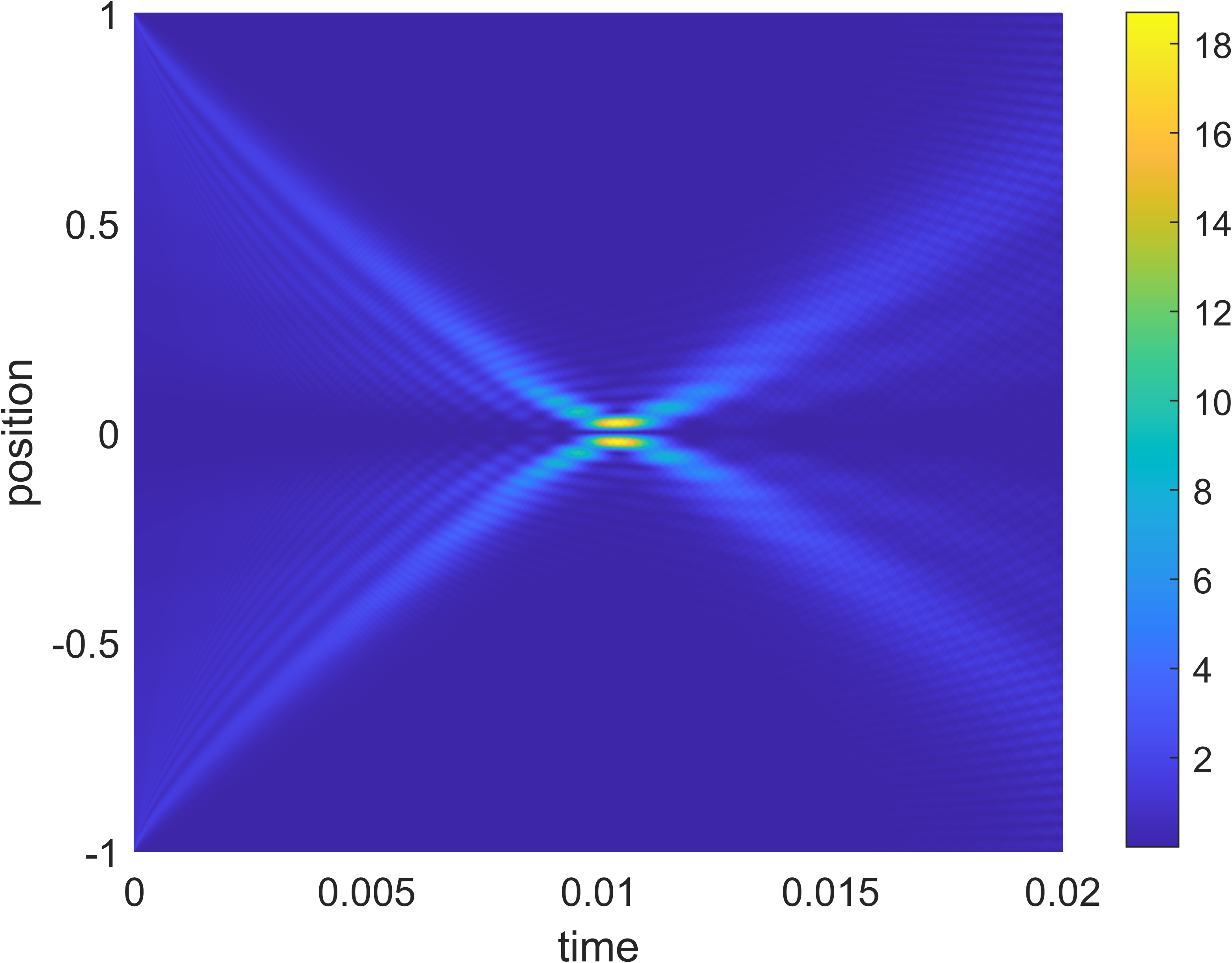}
\caption{Top view: Nodal eigenfunction } \label{subfig:nodweyll1t1top}
\end{subfigure}
	
\begin{subfigure}[b]{.32\linewidth}
\includegraphics[width=1\textwidth]{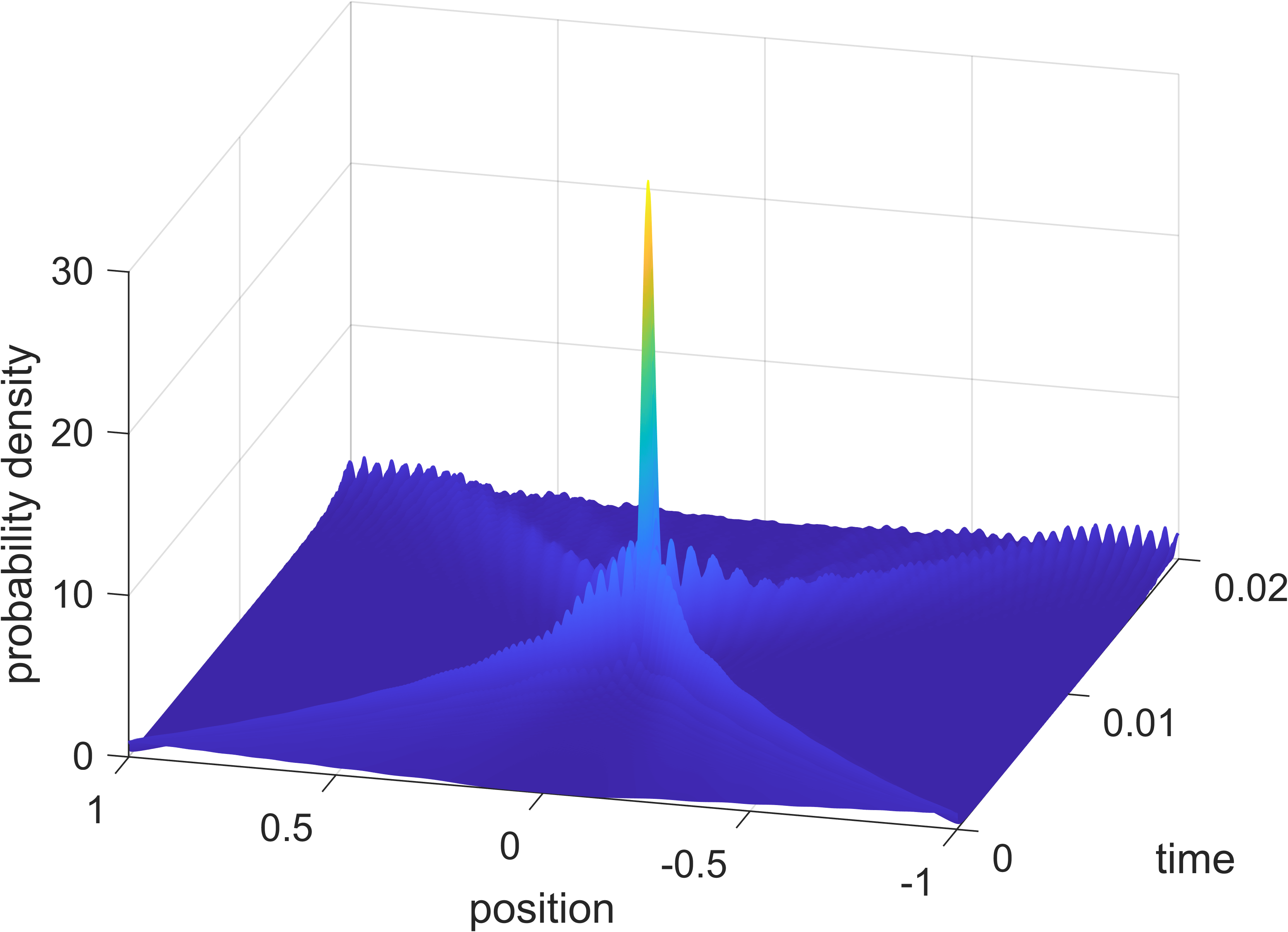}
\caption{Side view: Non-nodal eigenfunction } \label{subfig:nondweyll1t1side}
\end{subfigure}
\begin{subfigure}[b]{.32\linewidth}
\includegraphics[width=1\textwidth]{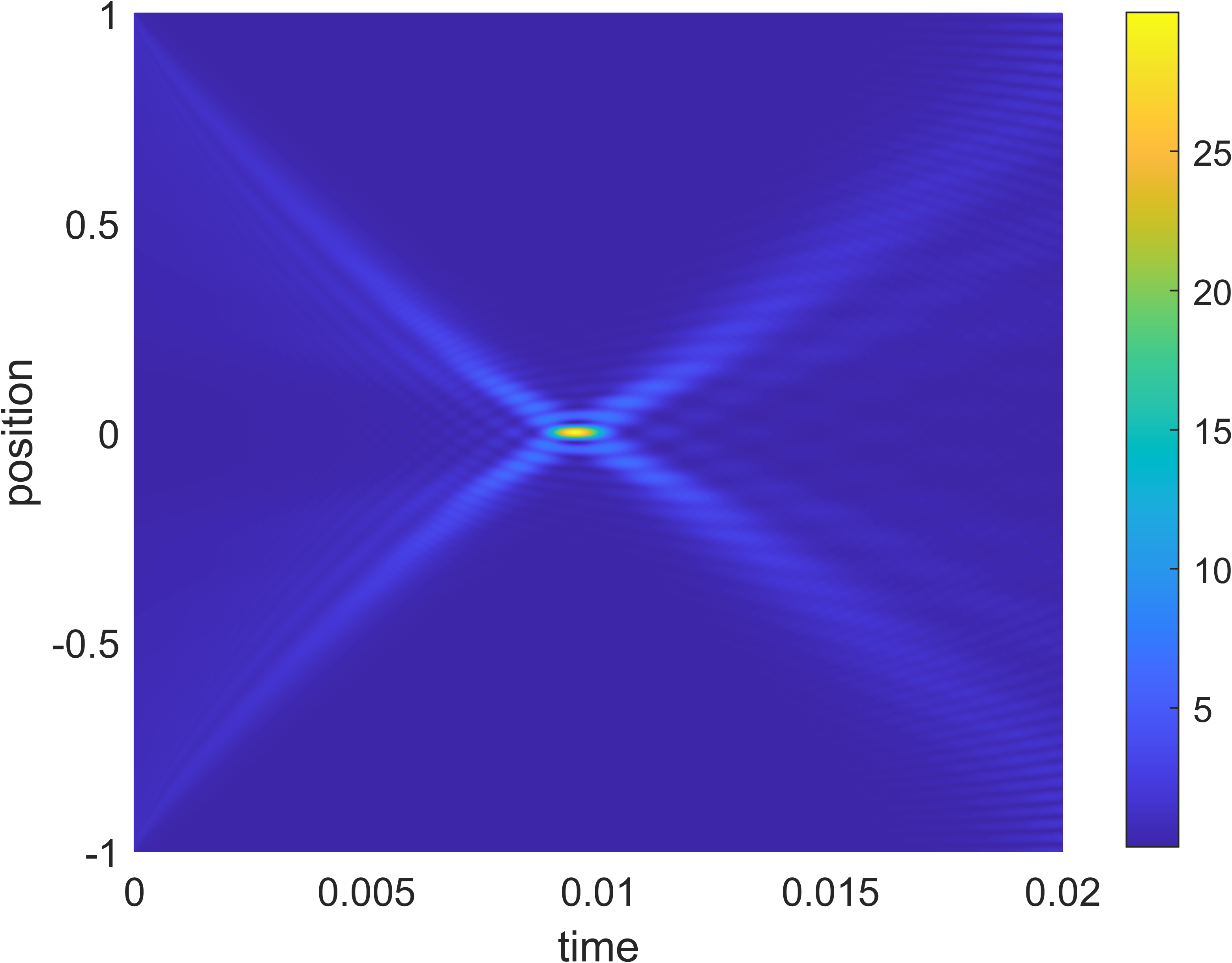}
\caption{Top view: Non-nodal eigenfunction } \label{subfig:nondweyll1t1top}
\end{subfigure}
	
\caption{The time evolution of the nodal and non-nodal eigenfunctions of the Weyl-quantized TOA operator for the cosine potential with confining length $l=1$ and eigenvalue $\tau=0.01$, respectively. The eigenfunctions also unitarily collapse at the arrival point $q=0$ at a time equal to the eigenvalue $\tau$, similar to the supraquantized operator.}\label{fig:weyll1tau1}
\end{figure}

\subsection{Comparison between supraquantized and Weyl-quantized operators for a cosine potential}

Let us now explore the difference between quantized and supraquantized TOA operators with respect to the dynamics of their eigenfunctions. For the cosine potential, the Weyl-quantized TOA operator assumes the form $\hat{\mathrm{T}}_W=(\mu/i \hbar)\,\int_{-\infty}^{\infty} \,dq' \,T_W(q,q') \,\mathrm{sgn}(q-q')\,\varphi(q')$, where the corresponding time kernel factor is given by
\begin{equation}\label{tweylcosine}
T_W(u,v)=\frac{1}{2}\int_{0}^{u} ds \,\, {}_0F_{1}\left(;1;\left(\frac{\mu\,V_0}{2\hbar^2}\right)v^2\left[\cos\left(\frac{ku}{2}\right)-\cos\left(\frac{kv}{2}\right)\right]\right),
\end{equation}
with $u=q+q'$ and $v=q-q'$. Since the potential being considered is nonlinear, the kernel factor $T_W(u,v)$ does not satisfy the TKE. Hence, the operator $\hat{\mathrm{T}}_W$ is not conjugate with the system Hamiltonian. Nonetheless, Figure (\ref{fig:weyll1tau1}) confirms that the Weyl-quantized operator also possesses the unitary arrival property for $V_0=k=l=1$. Within numerical accuracy, we do not see significant difference between the quantized and supraquantized operators. Note, however, that this does not mean the independence of the operator's dynamics on the time-energy canonical commutation relation.

In fact, we are able to distinguish the two operators when we increase the values of $l$, $V_0$ and $k$. For instance, Figure (\ref{fig:weylvssupr_data3_l3_tau2}) shows the time evolution of the nodal eigenfunctions of the Weyl-quantized and supraquantized TOA operators for the parameters $V_0=5, k=5$,  $l=3$ and $\tau=0.1$. Both operators still show nodal characteristics with two peaks merging at the arrival point $q=0$ at the eigenvalue time $\tau=0.1$. However, it is apparent that the the supraquantized TOA operator exhibits more desirable nodal dynamics compared to the Weyl-quantized one. Specifically, the supraquantized nodal eigenfunctions demonstrate sharper dynamics with less noise and taller peaks compared to the Weyl-quantized operator. 

On the other hand, Figure (\ref{fig:weylvssupr_data3_l3_tau2_noonn}) demonstrates the representative dynamics of the non-nodal eigenfunctions of the Weyl-quantized and supraquantized operators. We also observe a noisier and less sharp dynamics on the Weyl-quantized operator. In particular, the non-nodal eigenfunctions of the supraquantized operator has the smoothest and  sharpest probability density at the arrival point. Thus the supraquantized TOA operator demonstrate the most desirable non-nodal characteristics.

In Figures (\ref{fig:weylvssupr_data2_10_tau1_nod}) and (\ref{fig:weylvssupr_data2_10_tau1_noonn}), we set $V_0=5, k=1$,  $l=10$ and $\tau=0.1$. Notice that the eigenfunctions of the Weyl-quantized TOA operator already loses its nodal and non-nodal characteristics so that its unitary arrival property is no longer discernible. The supraquantized operator, on the other hand, still posses the desired nodal and non-nodal characteristics. 

One might wonder why there are no significant differences between the two operators when using the parameters (i) $V_0=k=l=1$, but there are noticeable distinctions when using the parameters (ii) $V_0=5, k=5, l=3$ and (iii) $V_0=5, k=1, l=10$. We can explain these behaviors using the expanded iterative solution of the TKE (\ref{tkesol3compact}). Recall that the supraquantized TOA operator for the cosine potential can be written as the expansion $\hat{\mathrm{T}}_S=\hat{\mathrm{T}}_W+\hat{\mathrm{T}}_1+\hat{\mathrm{T}}_2+...,$ where the leading term is the Weyl-quantized operator (\ref{tweylcosine}) and the succeeding terms are the quantum corrections which are already incorporated in the exact closed-form solution. 

In the case where $V_0=k=l=1$, both operators behave similarly because their time kernel factors align for small parameter values. This implies that the quantum corrections  $\hat{\mathrm{T}}_n$ are negligible. In this way, even though the Weyl-quantized operator does not precisely satisfy the time-energy canonical commutation relation, it still serves as an approximate solution of the said conjugacy relation. Hence, it exhibits similar dynamics with that of the conjugacy-preserving TOA operator. However, when $V_0=5, k=5, l=3$, the quantum corrections become significant so that the Weyl-quantized operator alone is no longer a sufficient approximation for the supraquantized one. In the third case where $V_0=5, k=1, l=10$, we see that the representative dynamics of the Weyl-quantized operator is completely different with that of the supraquantized operator and far from the ideal dynamics. This can be explained in two ways. First, the quantum corrections already dominate and dictate the general behavior of the conjugacy-preserving operator, so that without them, no arrival property is observed. Another possible explanation is that the numerical algorithm used is no longer suitable for the Weyl-quantized TOA operator so that the corresponding dynamics already diverge. Whatever the case may be, both scenarios indicate that the supraquantized TOA operator comes closest to the ideal unitary arrival property among arrival time observables within our numerical algorithm. 

These findings shed light on a crucial aspect of the quantum corrections $\hat{\mathrm{T}}_n$, not discussed in Ref. \cite{Pablico2023}. The noisy feature seen in the dynamics of the Weyl-quantized operator result from the absence of the quantum corrections. That is, the $\hat{\mathrm{T}}_n$ operators ensures smooth nodal and non-nodal characteristics of arrival time operators. Put simply, they eliminate the noise that shows up in the behavior of the Weyl-quantized operator. Consequently, the supraquantized operator predicts with greater certainty the arrival and non-arrival of the particle at the arrival point.

Our result is also significant because it helps us understand the role of the time-energy canonical commutation relation in the observed dynamics of time operators. The conjugacy relation with the system Hamiltonian ensures that the TOA operator evolves in step with parametric time, leading to precise and localized arrival of eigenfunctions at the designated arrival point at the corresponding eigenvalue time. Conversely, non-conjugacy with the Hamiltonian can lead to two scenarios. Firstly, a non-conjugacy-preserving TOA operator which is also a sufficient approximation to the supraquantized operator, still possess the unitary arrival properties but lack the sharpness and precision of the latter operator. Secondly, a non-conjugacy-preserving TOA operator, which is also a poor approximation to the supraquantized operator, results in the eigenfunction's arrival not coinciding with its eigenvalue. The latter scenario has also been observed by two of us in Ref. \cite{Galapon2018} where nonunitary collapse occurs for deformed-quantized operators constructed by multiplying some deformation factor to the Weyl-quantized operator. 

Now, in all examples examined, both the supraquantized and Weyl-quantized operators demonstrate the unitary arrival property. However, this property is more consistently observed, within numerical accuracy, with respect to the nodal and non-nodal eigenfunctions of the supraquantized TOA operator for varying values of $l$, $V_0$ and $k$, whether small or large. Therefore, we conclude that the supraquantized (conjugacy-preserving) TOA operator exhibits more desirable dynamics compared to the Weyl-quantized (non-conjugacy preserving) TOA operator. This also implies that the time-energy canonical commutation relation plays a crucial role in shaping the dynamics of time operators. Consequently, the quantum corrections appearing in the expanded iterative solution of the TKE are necessary to achieve sharper arrivals at the designated arrival point. The same results have been observed for a sinusoidal potential defined by $V(q)=V_0\,\mathrm{sin}(kq).$ Hence, our third and final objective is fulfilled. 


\begin{figure}
\centering
\begin{subfigure}[b]{.32\linewidth}
\includegraphics[width=1\textwidth]{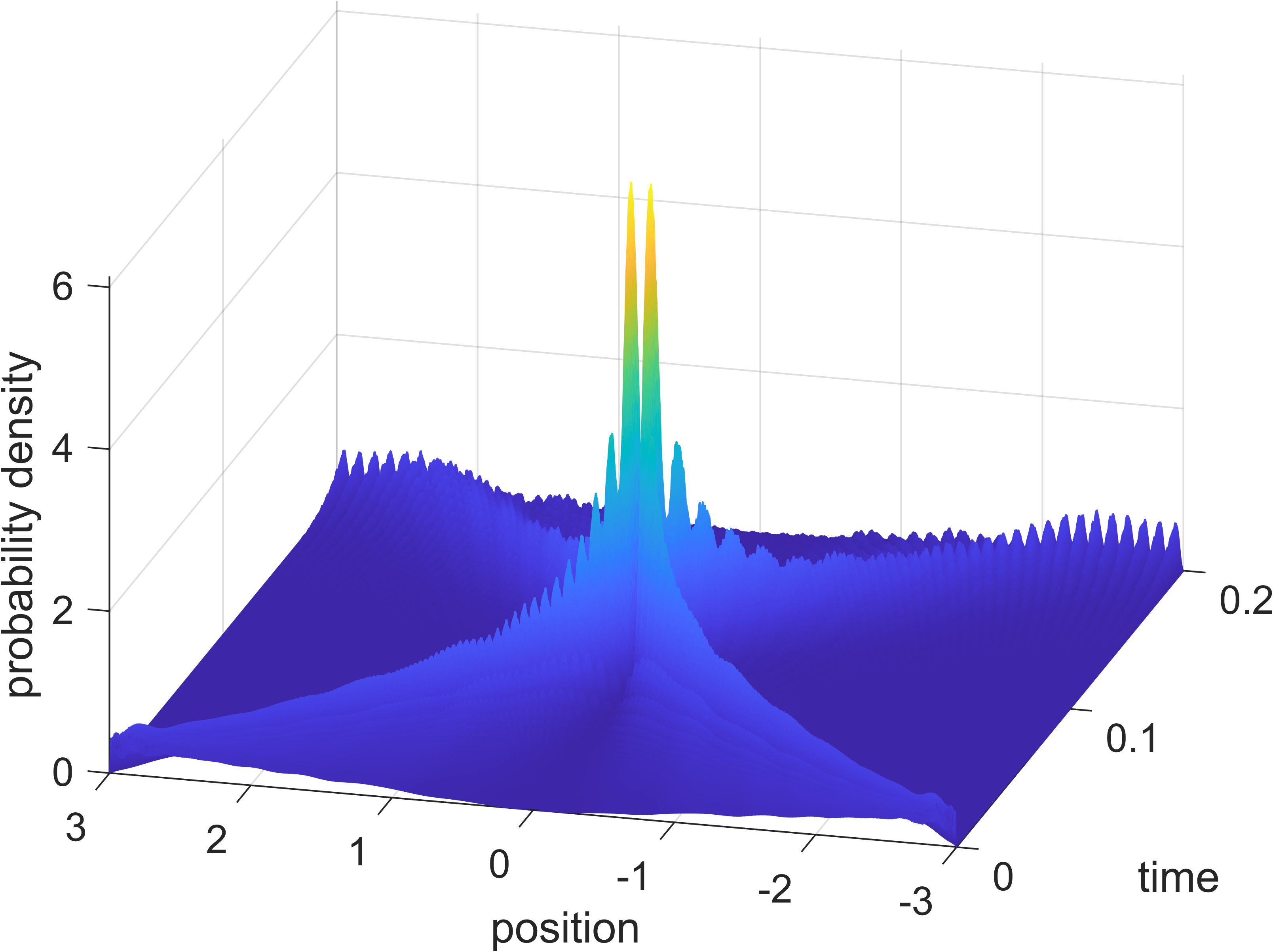}
\caption{Side view: Supraquantized nodal eigenfunction} \label{subfig:supra_data3_l3_tau2_side}
\hspace{20mm}
\end{subfigure}
\begin{subfigure}[b]{.32\linewidth}
\includegraphics[width=1\textwidth]{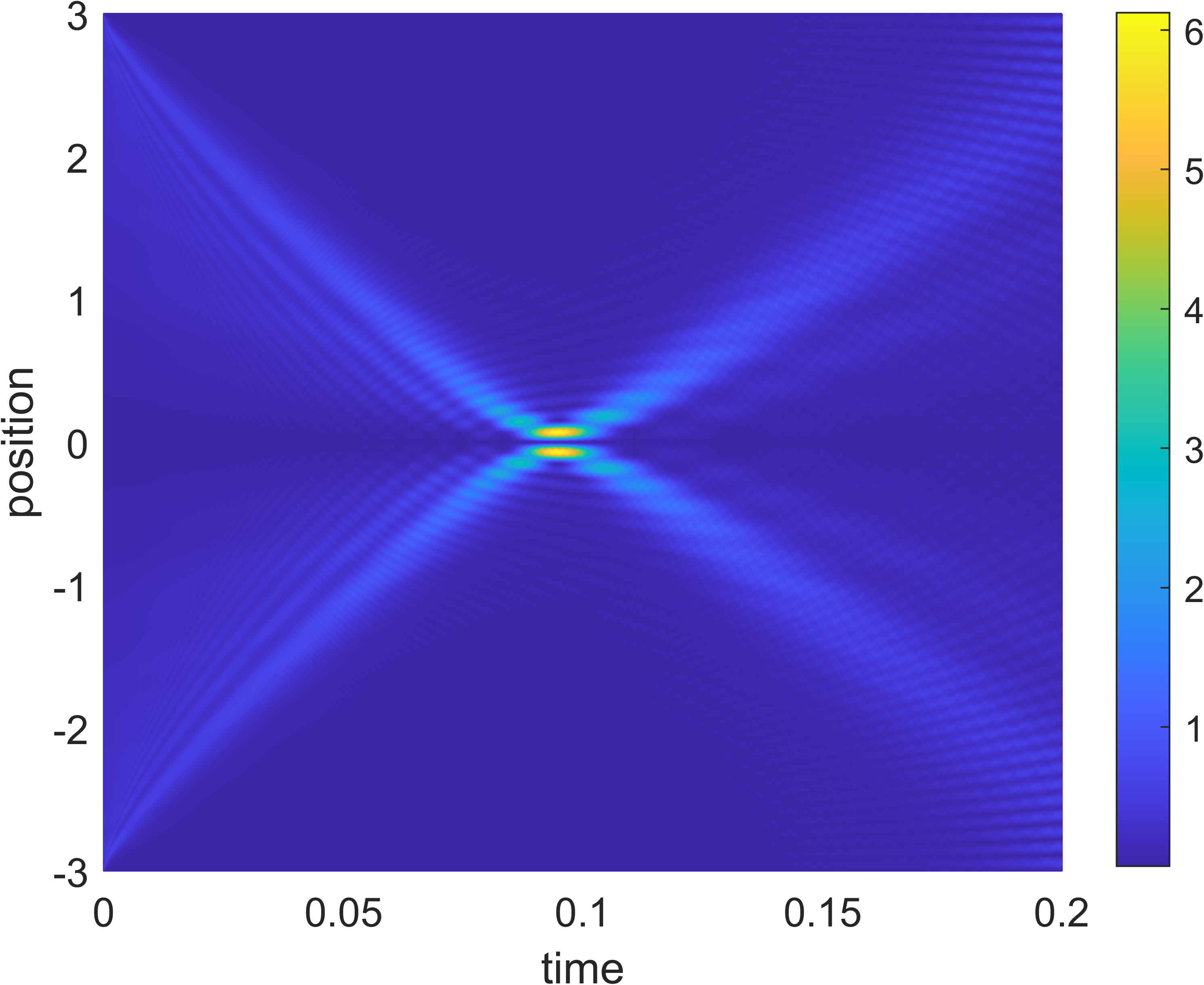}
\caption{Top view: Supraquantized nodal eigenfunction} \label{subfig:supra_data3_l3_tau2_top}
\hspace*{20mm}
\end{subfigure}
	
\begin{subfigure}[b]{.32\linewidth}
\includegraphics[width=1\textwidth]{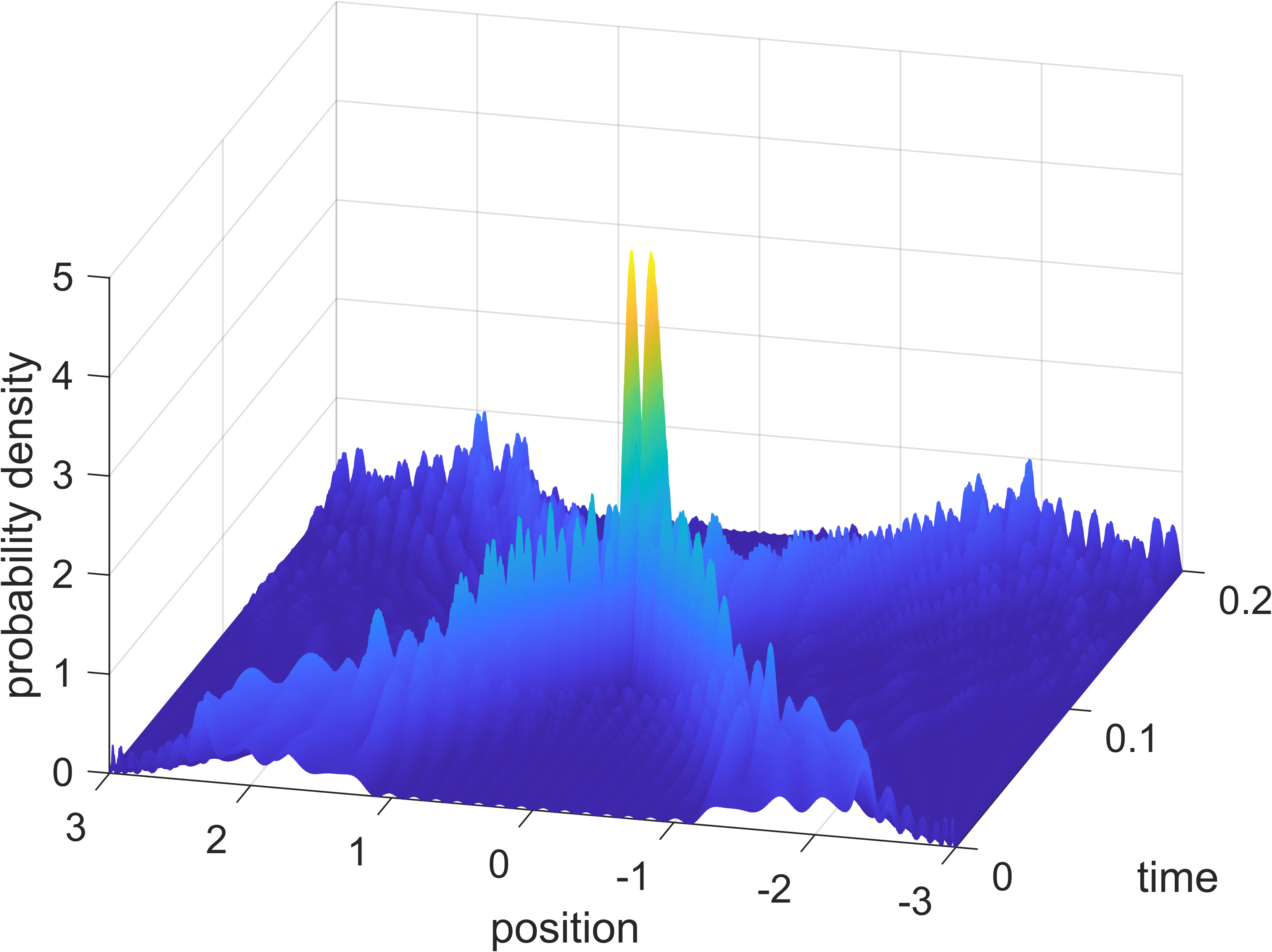}
\caption{Side view: Weyl-quantized nodal eigenfunction} \label{weyl:supra_data3_l3_tau2_side}
\end{subfigure}
\begin{subfigure}[b]{.32\linewidth}
\includegraphics[width=1\textwidth]{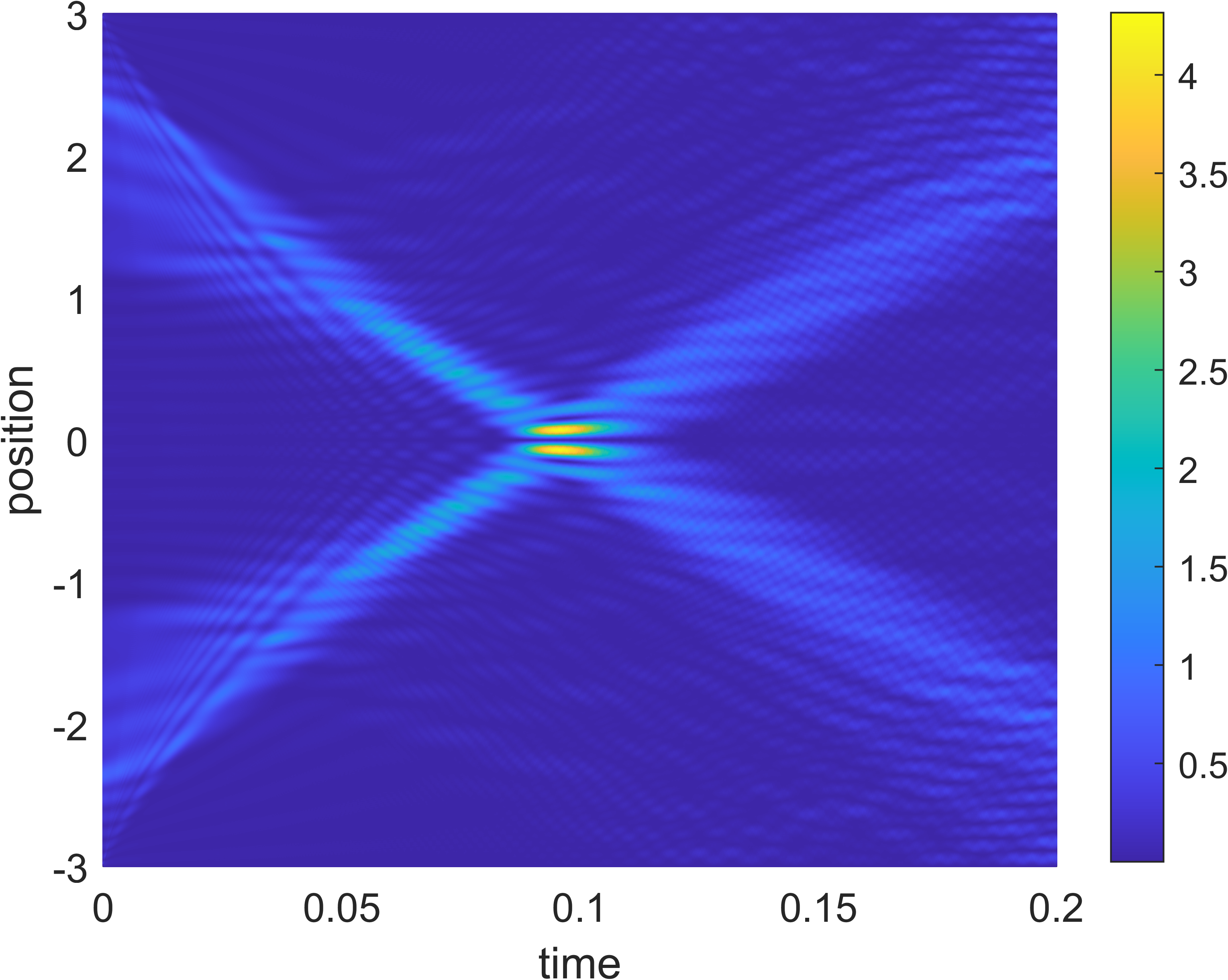}
\caption{Top view: Weyl-quantized nodal eigenfunction } \label{subfig:weyl_data3_l3_tau2_top}
\end{subfigure}
	
\caption{The time evolution of the nodal eigenfunctions of the (a)--(b) supraquantized and (c)--(d) Weyl-quantized TOA operator for the cosine potential with parameters $V_0=5, k=5$, confining length $l=3$ and eigenvalue $\tau=0.1$. The supraquantized TOA operator exhibits the most ideal unitary dynamics.}\label{fig:weylvssupr_data3_l3_tau2}
\end{figure}


\begin{figure}
\centering
\begin{subfigure}[b]{.32\linewidth}
\includegraphics[width=1\textwidth]{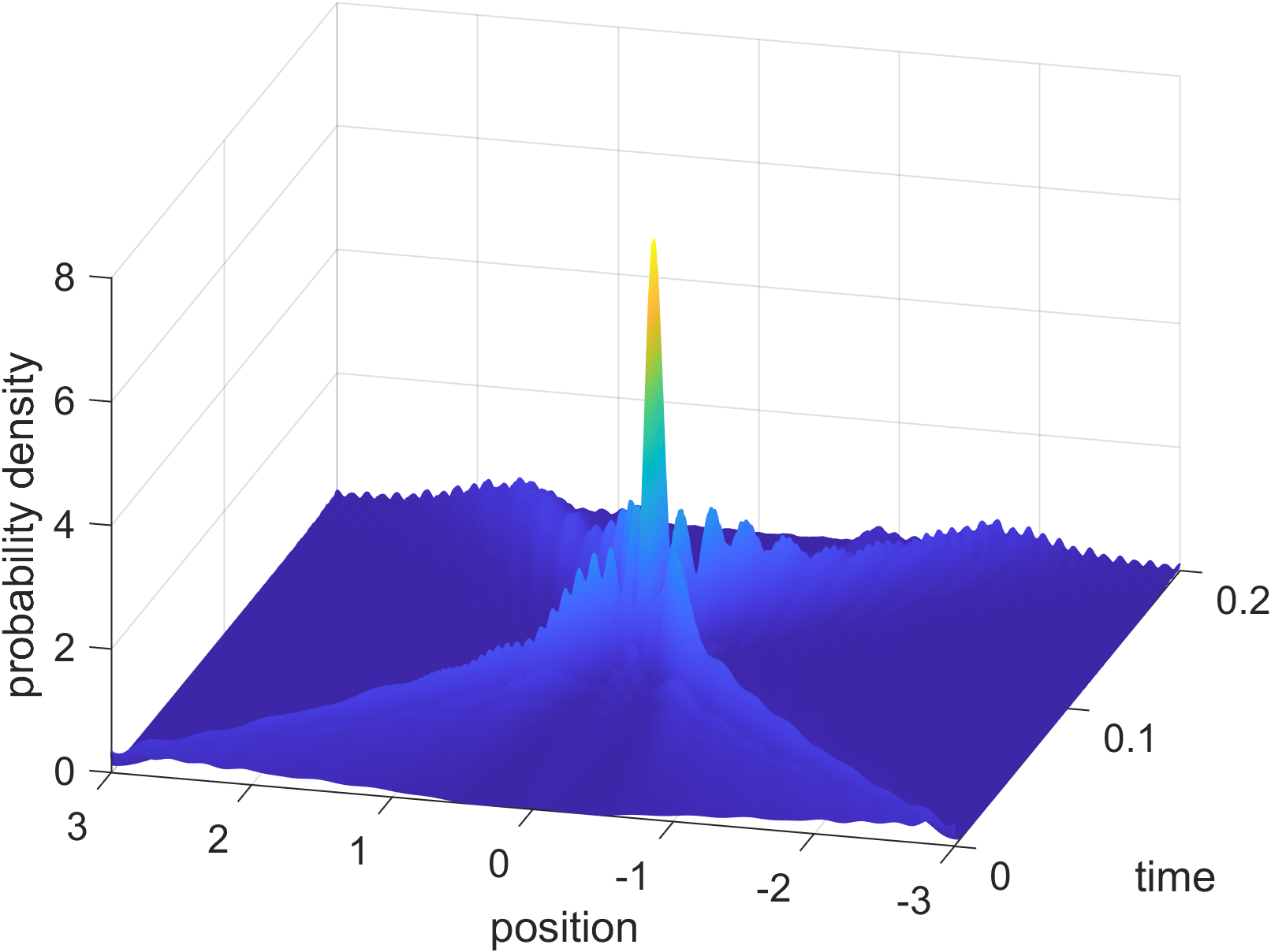}
\caption{Side view: Supraquantized non-nodal eigenfunction} \label{subfig:supra_data3_l3_tau2_side_nonn}
\end{subfigure}
\begin{subfigure}[b]{.32\linewidth}
\includegraphics[width=1\textwidth]{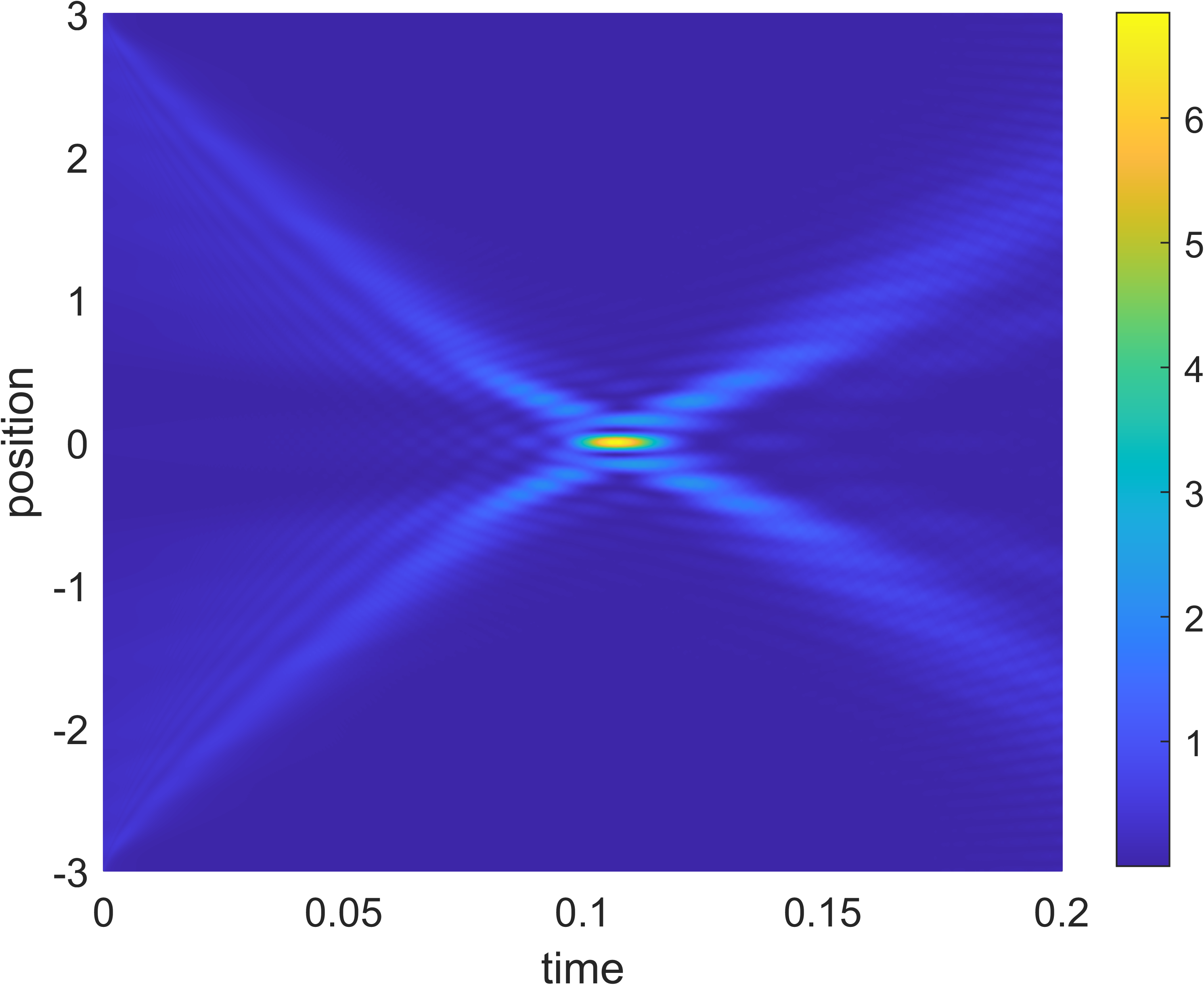}
\caption{Top view: Supraquantized non-nodal eigenfunction} \label{subfig:supra_data3_l3_tau2_top_nonn}
\end{subfigure}
	
\begin{subfigure}[b]{.32\linewidth}
\includegraphics[width=1\textwidth]{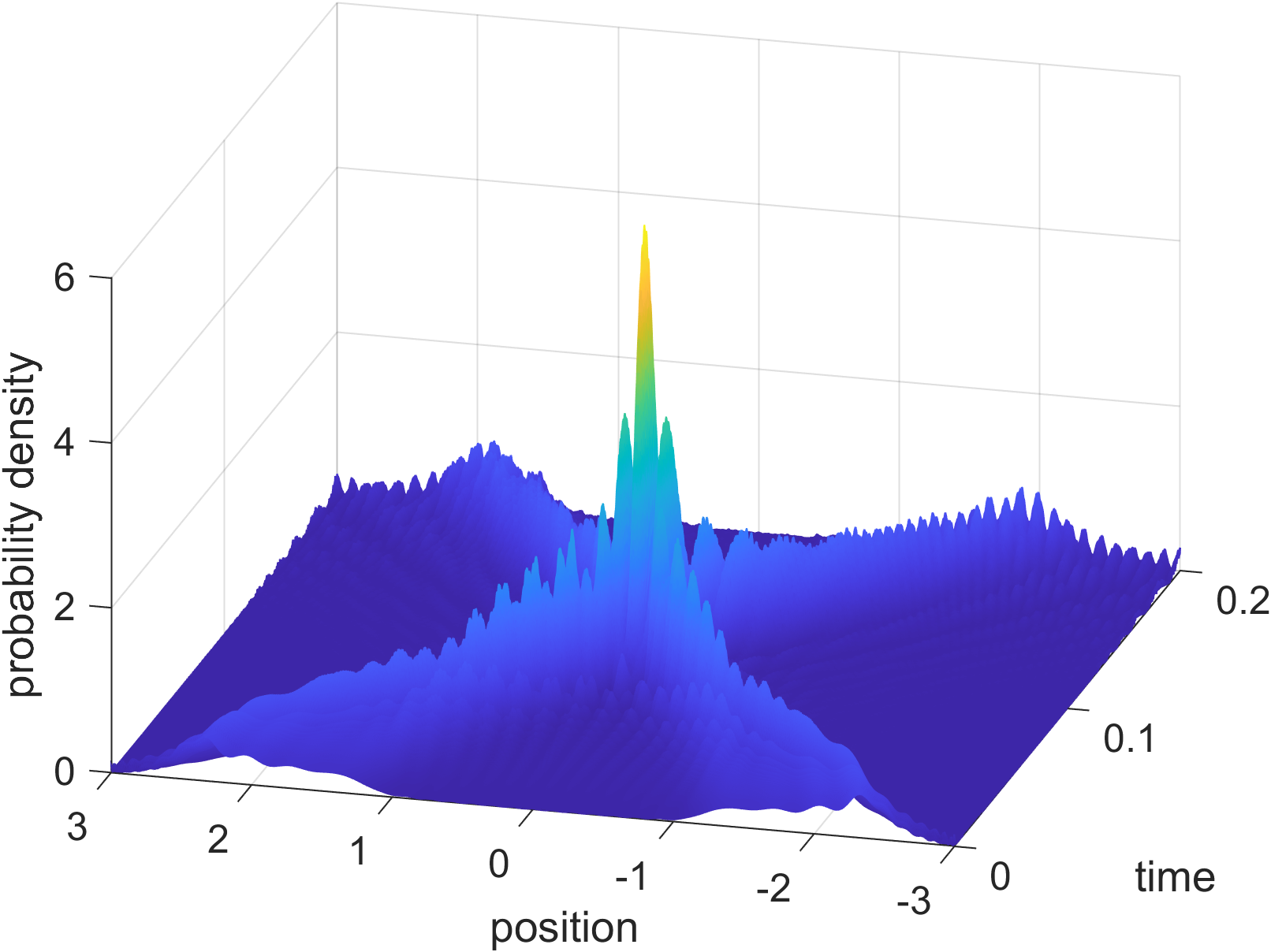}
\caption{Side view: Weyl-quantized non-nodal eigenfunction} \label{weyl:supra_data3_l3_tau2_side_nonn}
\end{subfigure}
\begin{subfigure}[b]{.32\linewidth}
\includegraphics[width=1\textwidth]{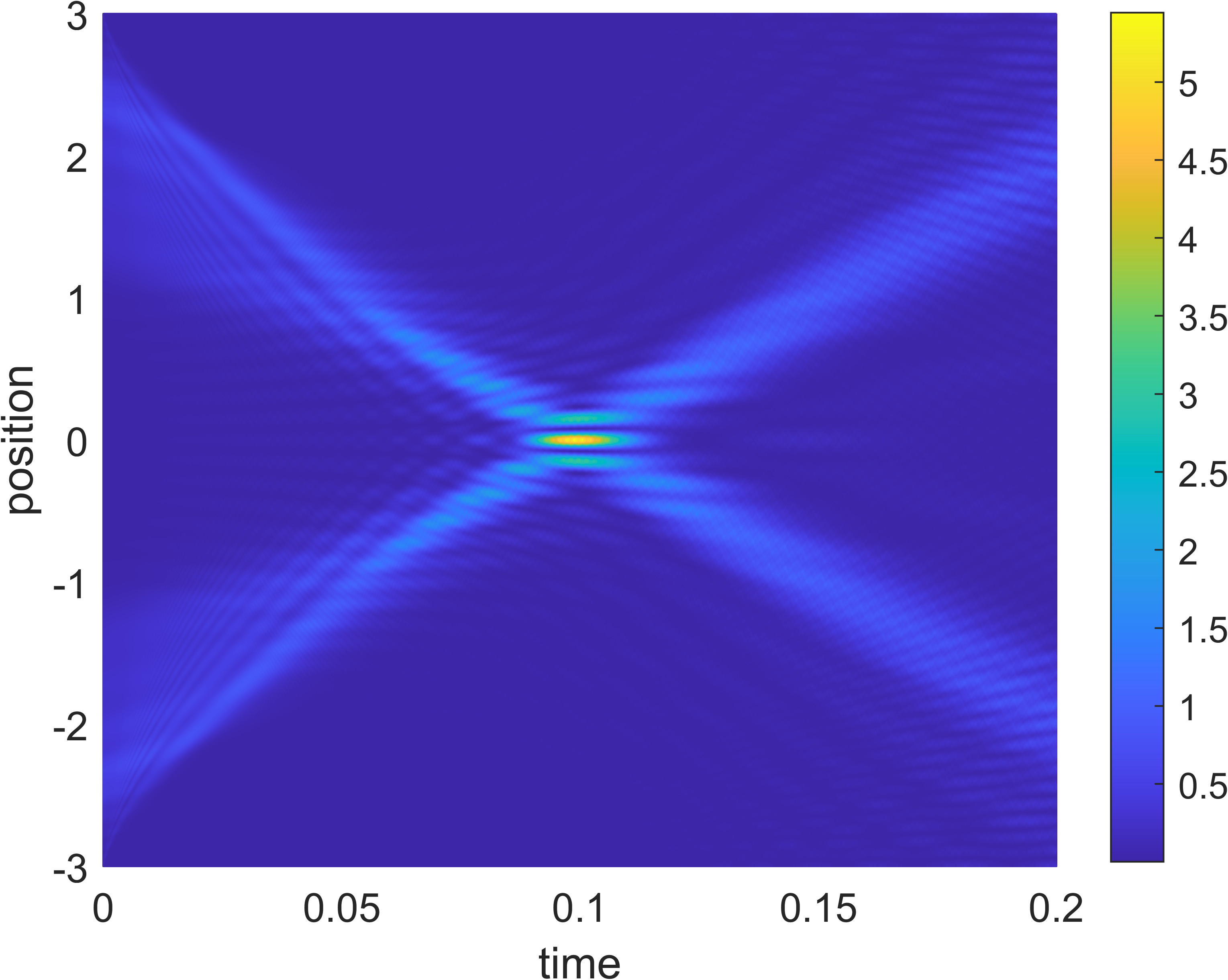}
\caption{Top view: Weyl-quantized non-nodal eigenfunction } \label{subfig:weyl_data3_l3_tau2_top_nonn}
\end{subfigure}
	
\caption{The time evolution of the non-nodal eigenfunctions of the (a)--(b) supraquantized and (c)--(d) Weyl-quantized TOA operator for the cosine potential with parameters $V_0=5, k=5$, confining length $l=3$ and eigenvalue $\tau=0.1$. The supraquantized TOA operator exhibits the most ideal unitary dynamics.}\label{fig:weylvssupr_data3_l3_tau2_noonn}
\end{figure}


\begin{figure}
\centering
\begin{subfigure}[b]{.32\linewidth}
\includegraphics[width=1\textwidth]{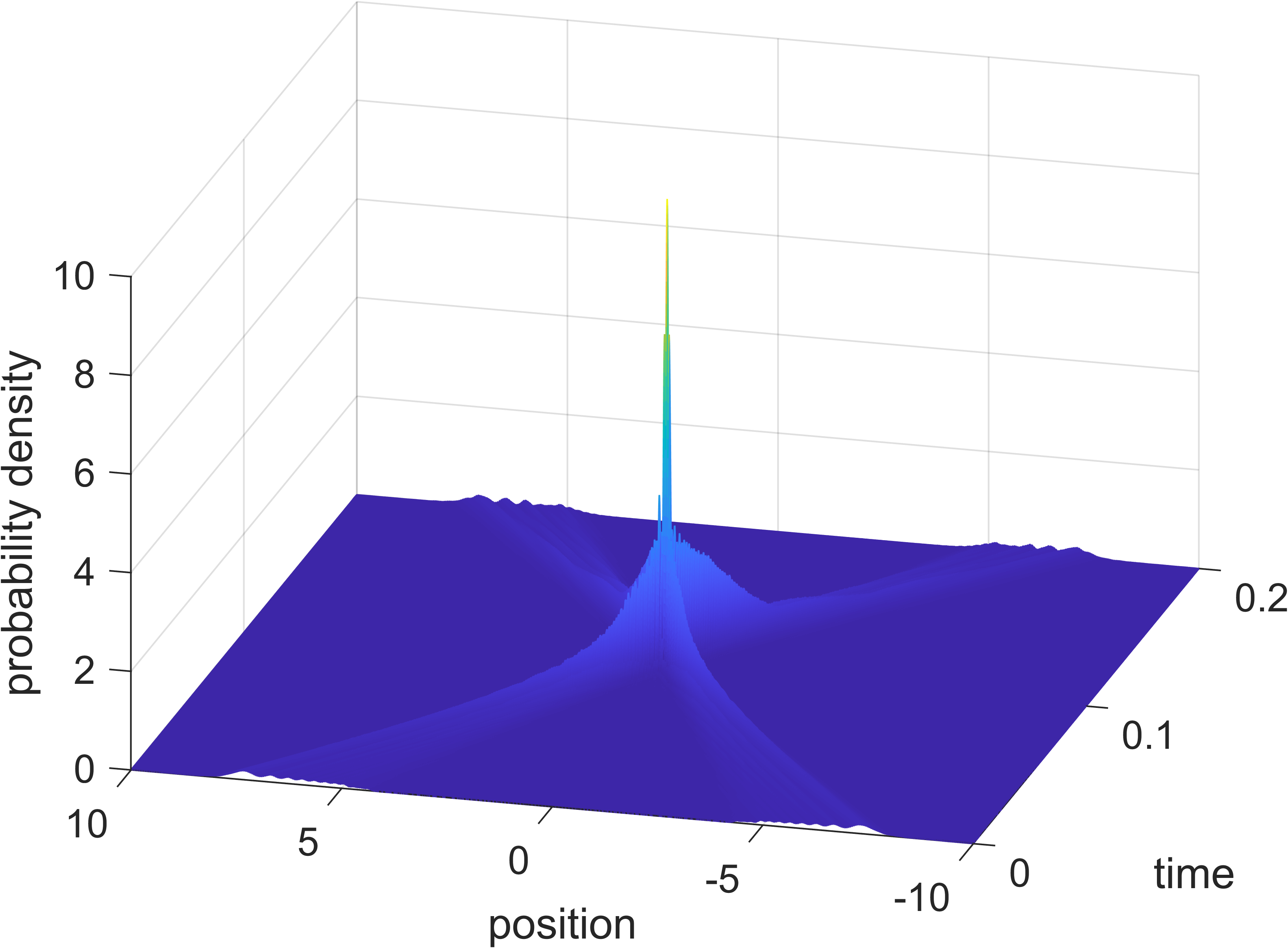}
\caption{Side view: Supraquantized non-nodal eigenfunction} \label{subfig:supra_data2_l10_tau1_side_nod}
\end{subfigure}
\begin{subfigure}[b]{.32\linewidth}
\includegraphics[width=1\textwidth]{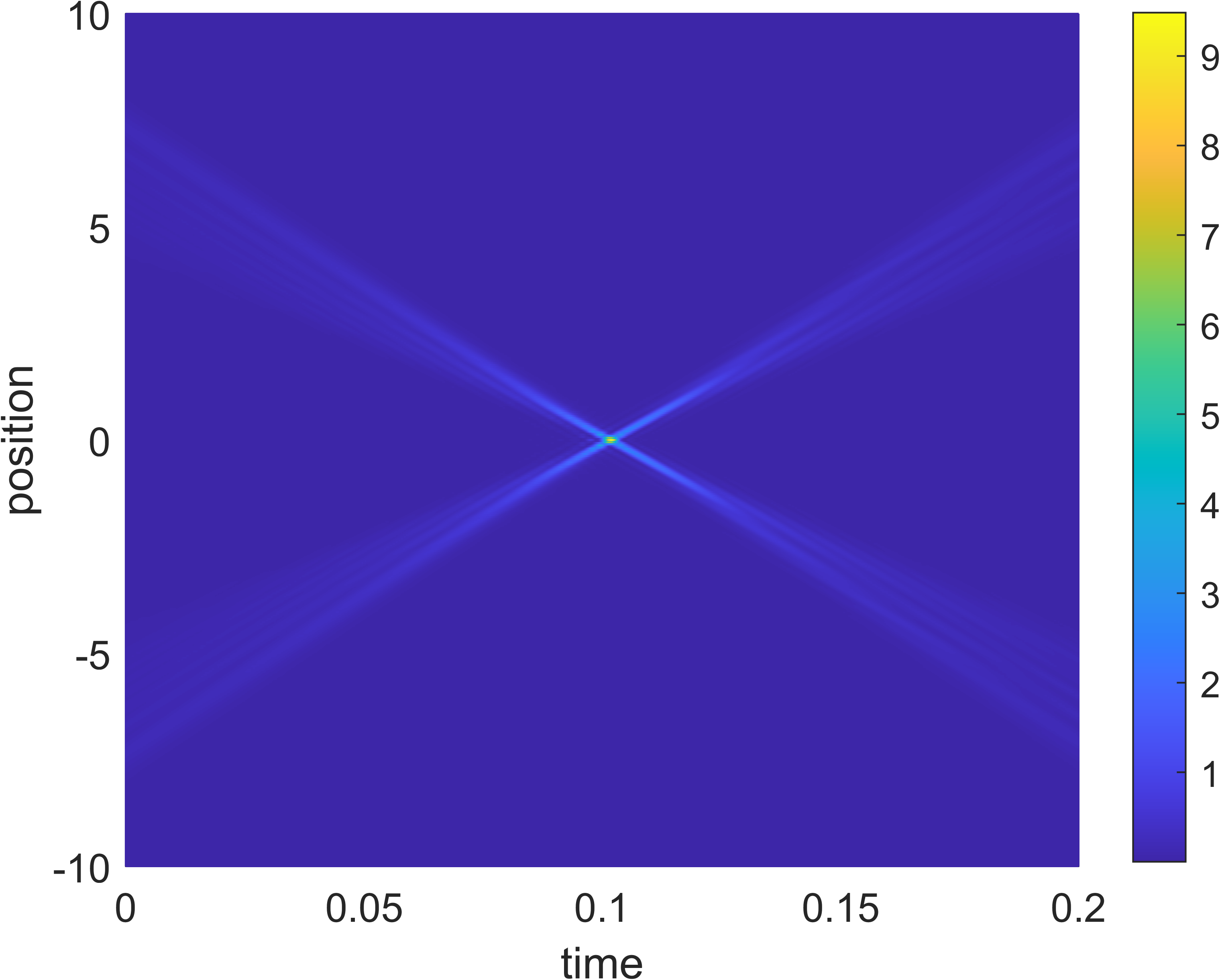}
\caption{Top view: Supraquantized non-nodal eigenfunction} \label{subfig:supra_data2_l10_tau1_top_nod}
\end{subfigure}
	
\begin{subfigure}[b]{.32\linewidth}
\includegraphics[width=1\textwidth]{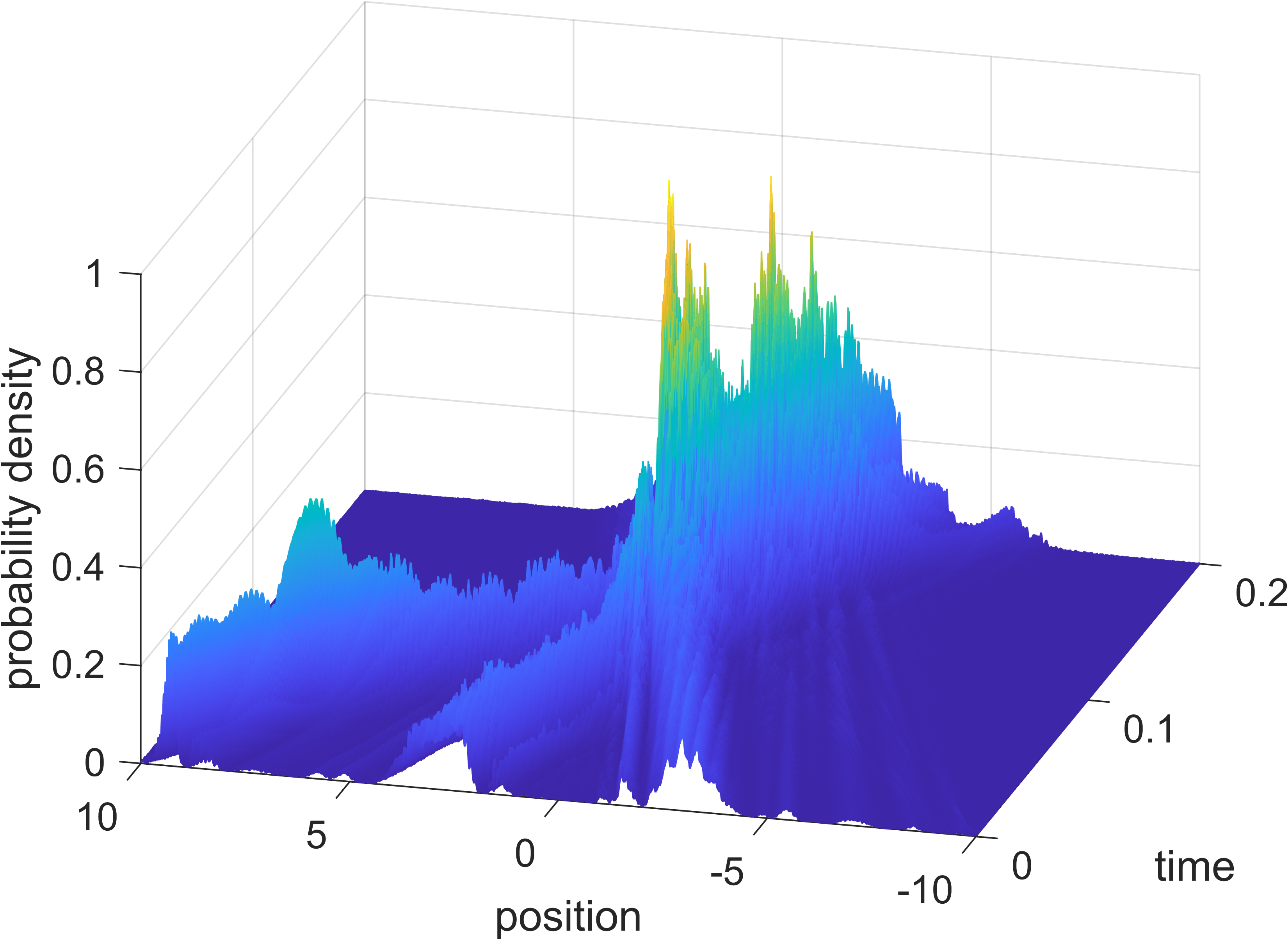}
\caption{Side view: Weyl-quantized non-nodal eigenfunction} \label{weyl:supra_data2_l10_tau1_side_nod}
\end{subfigure}
\begin{subfigure}[b]{.32\linewidth}
\includegraphics[width=1\textwidth]{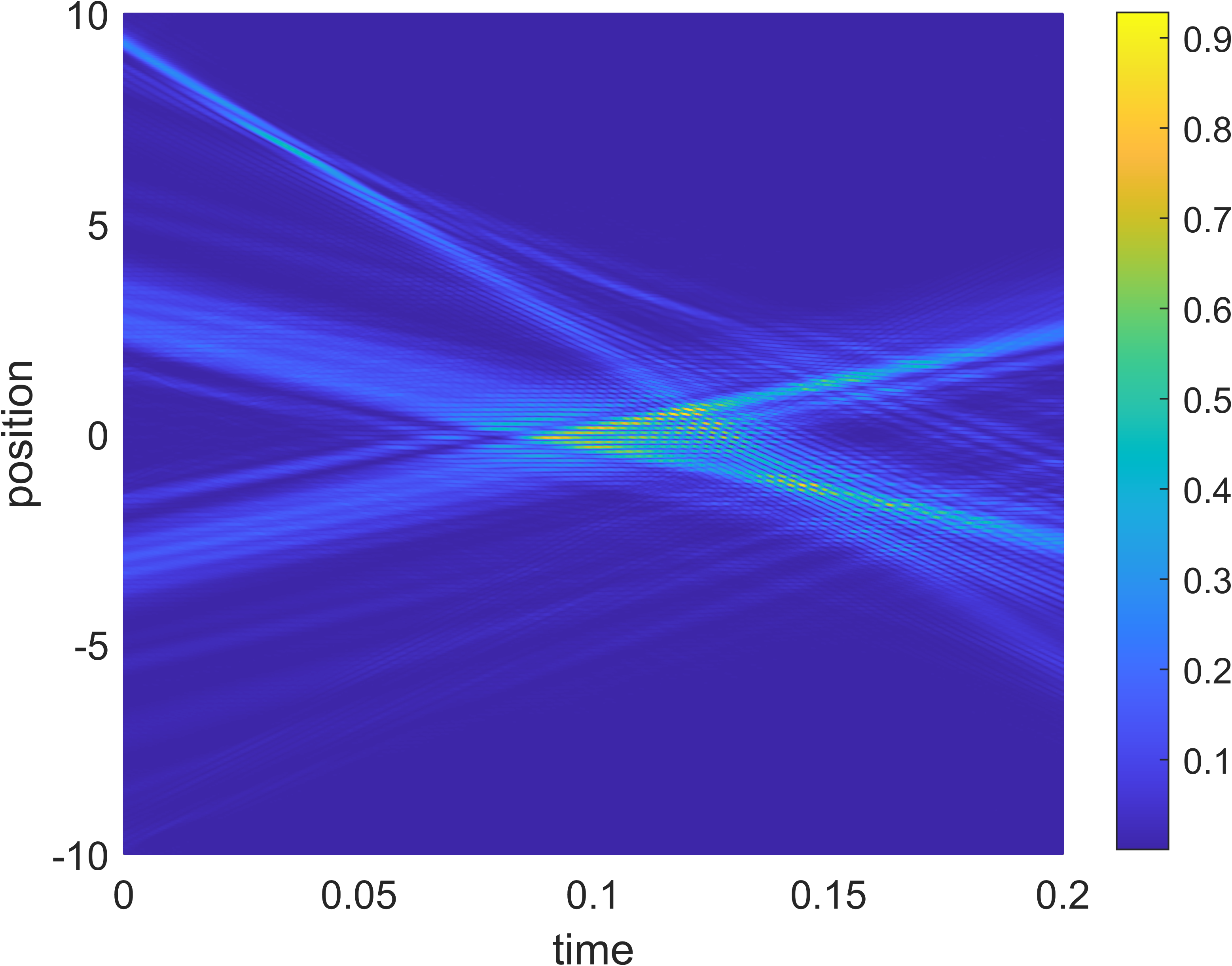}
\caption{Top view: Weyl-quantized non-nodal eigenfunction } \label{subfig:weyl_data2_l10_tau1_top_nod}
\end{subfigure}
	
\caption{The time evolution of the non-nodal eigenfunctions of the (a)--(b) supraquantized and (c)--(d) Weyl-quantized TOA operator for the cosine potential with parameters $V_0=5, k=1$, confining length $l=10$ and eigenvalue $\tau=0.1$. The supraquantized TOA operator exhibits the most ideal unitary dynamics.}\label{fig:weylvssupr_data2_10_tau1_nod}
\end{figure}


\begin{figure}
\centering
\begin{subfigure}[b]{.32\linewidth}
\includegraphics[width=1\textwidth]{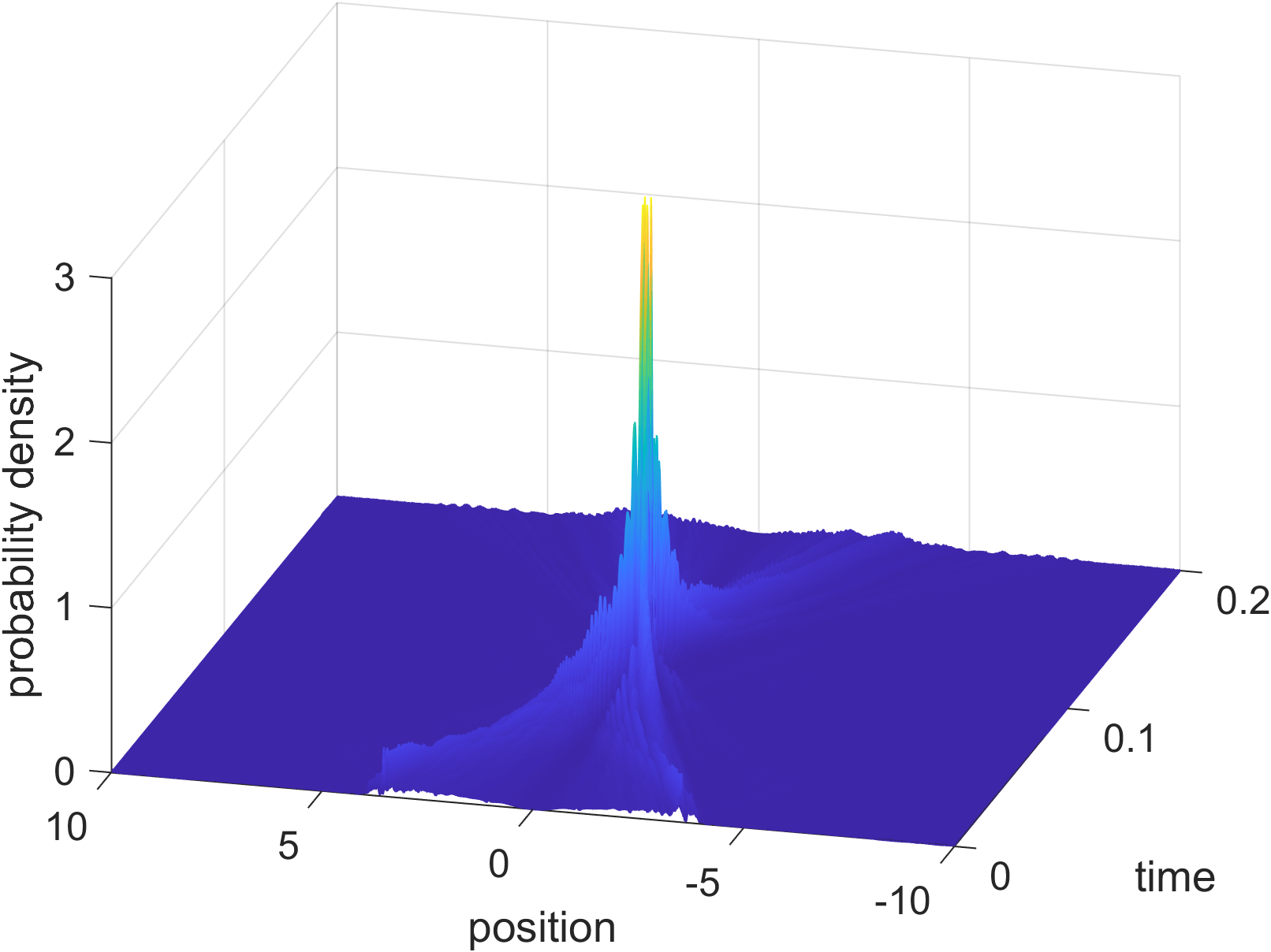}
\caption{Side view: Supraquantized nodal eigenfunction} \label{subfig:supra_data2_l10_tau1_side_nonn}
\end{subfigure}
\begin{subfigure}[b]{.32\linewidth}
\includegraphics[width=1\textwidth]{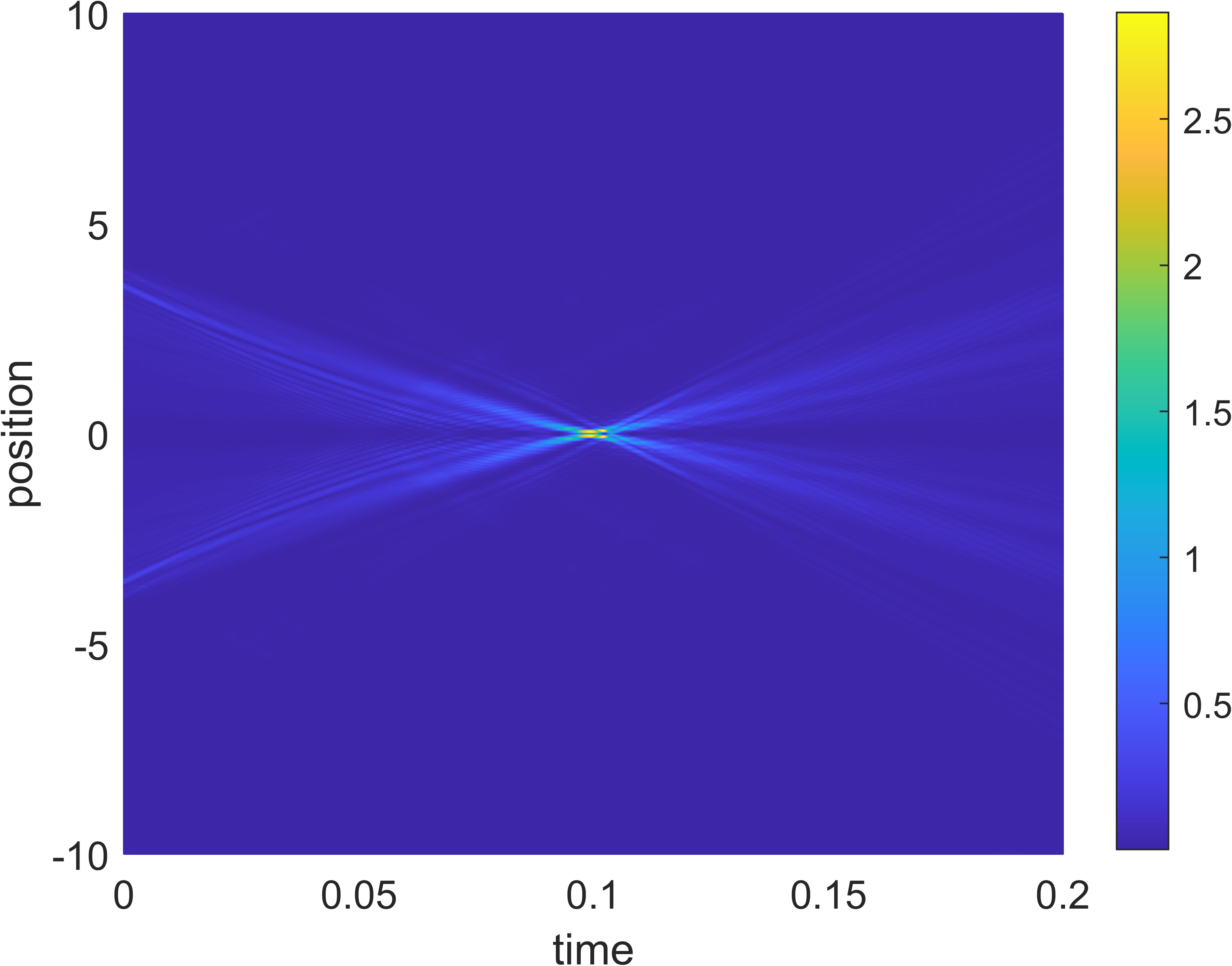}
\caption{Top view: Supraquantized nodal eigenfunction} \label{subfig:supra_data2_l10_tau1_top_nonn}
\end{subfigure}
	
\begin{subfigure}[b]{.32\linewidth}
\includegraphics[width=1\textwidth]{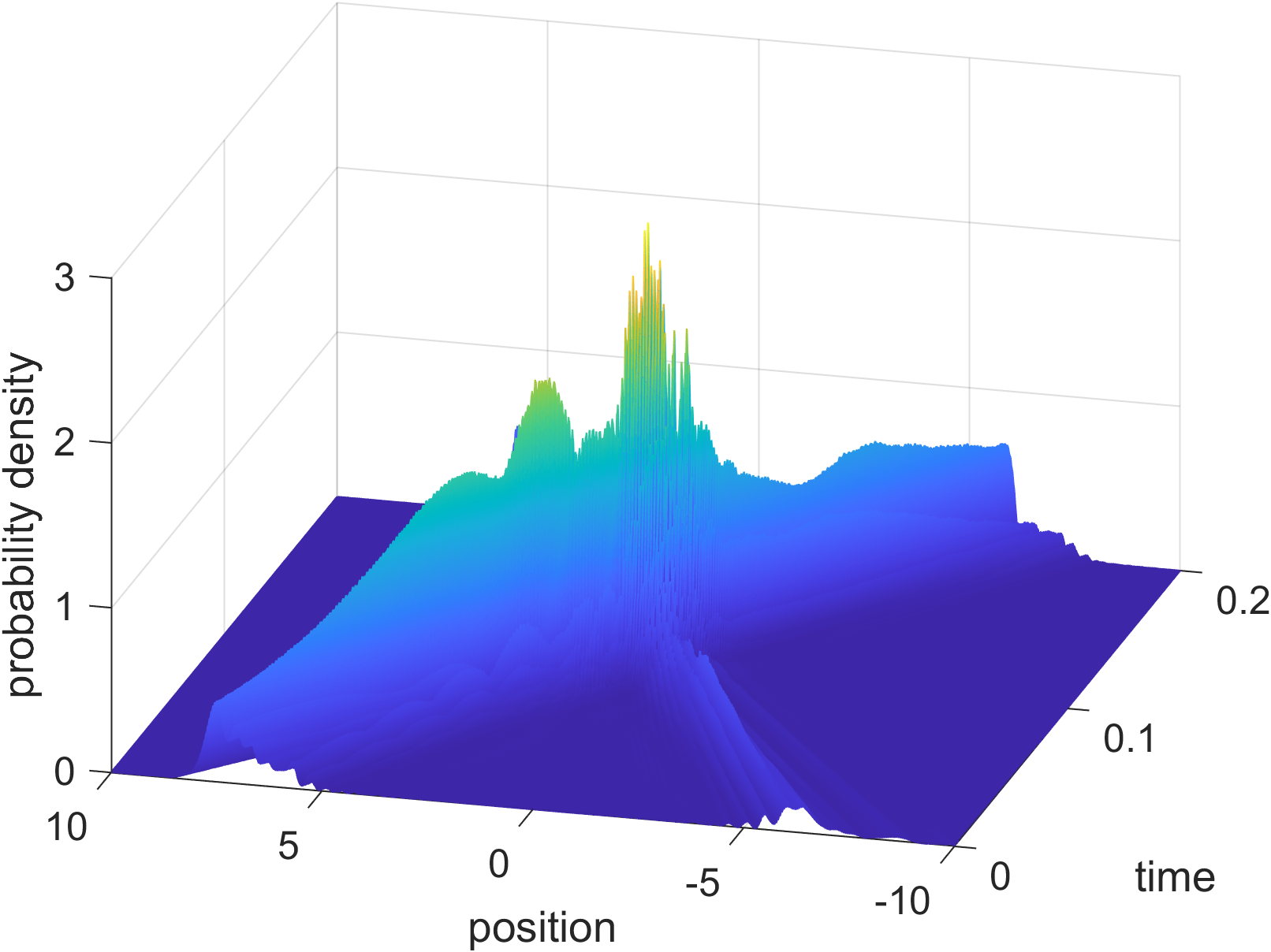}
\caption{Side view: Weyl-quantized nodal eigenfunction} \label{weyl:supra_data2_l10_tau1_side_nonn}
\end{subfigure}
\begin{subfigure}[b]{.32\linewidth}
\includegraphics[width=1\textwidth]{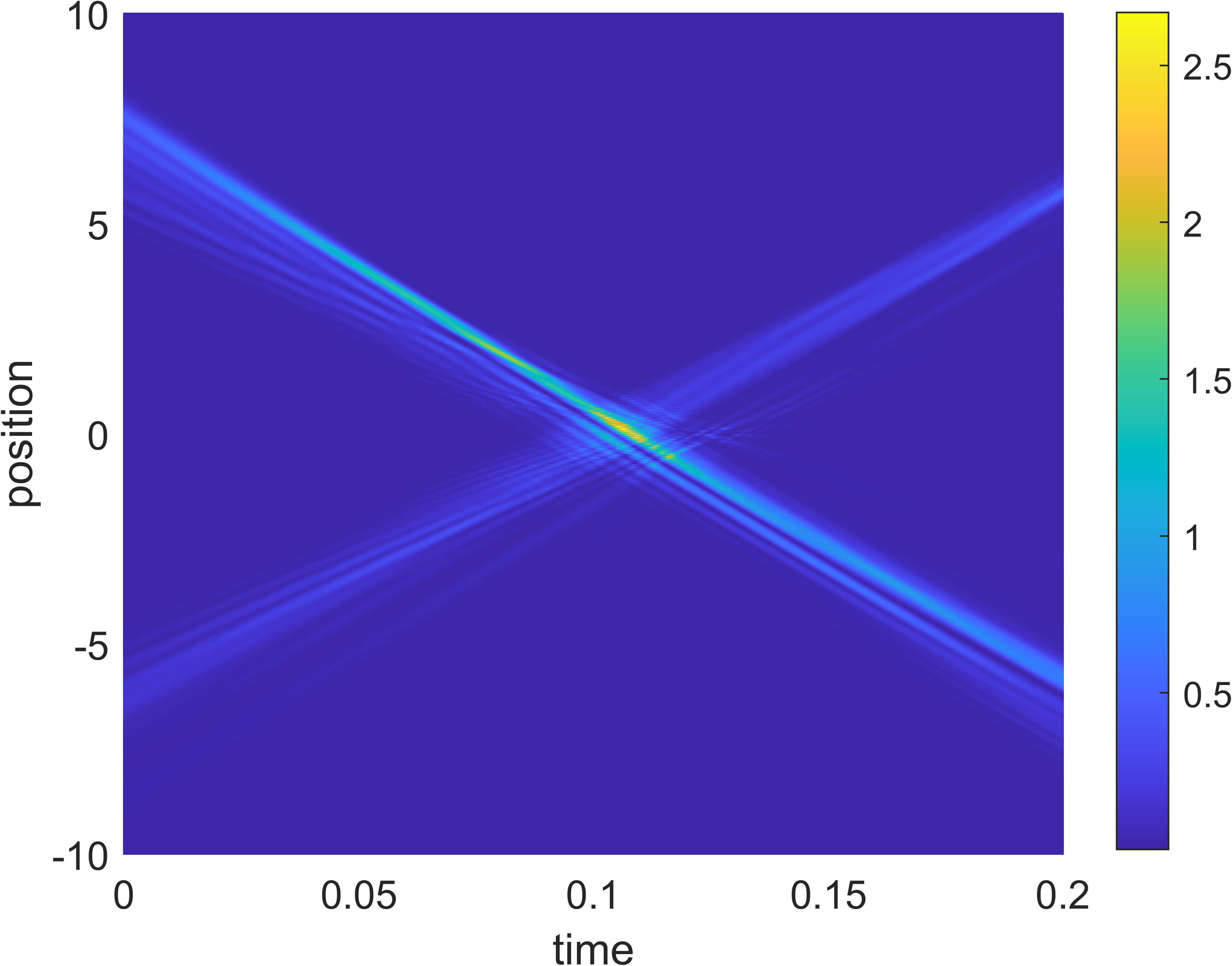}
\caption{Top view: Weyl-quantized nodal eigenfunction } \label{subfig:weyl_data2_l10_tau1_top_nonn}
\end{subfigure}
	
\caption{The time evolution of the \sout{non-}nodal eigenfunctions of the (a)--(b) supraquantized and (c)--(d) Weyl-quantized TOA operator for the cosine potential with parameters $V_0=5, k=1$, confining length $l=10$ and eigenvalue $\tau=0.1$. The supraquantized TOA operator exhibits the most ideal unitary dynamics.}\label{fig:weylvssupr_data2_10_tau1_noonn}
\end{figure}

\section{Conclusions} \label{sec:conclusion}

In this paper, we addressed the problem of constructing TOA operators canonically conjugate to the system Hamiltonian with separable potentials without relying on canonical quantization methods. Our construction entails finding an exact-closed form solution to a second-order partial differential equation, called the time kernel equation. This equation arises directly from the conjugacy requirement between our TOA operator and system Hamiltonian. 

Our exact solution enabled us to explore the properties of the conjugacy-preserving TOA operator. Specifically, we have shown that the constructed operator satisfies essential properties of a TOA observable, such as hermiticity, time-reversal symmetry, and correspondence between classical and quantum observables. We have also demonstrated it is a physically meaningful TOA observable by showing that its eigenfunctions exhibit the desired unitary arrival property at the arrival point equal to its respective eigenvalue. 

Furthermore, we compared the dynamics of the conjugacy-preserving (supraquantized) TOA and non-conjugacy-preserving (Weyl-quantized) operators with regards to the behavior of their nodal and non-nodal eigenfunctions. Nodal eigenfunctions correspond to particle appearance at the arrival point without detection, whereas non-nodal eigenfunctions represent particle appearance with detection. We have showed that the supraquantized operator consistently demonstrates more ideal nodal and non-nodal eigenfunctions within numerical accuracy. On the other hand, we demonstrated specific cases where the unitary dynamics of the Weyl-quantized operator lack the sharpness and precision exhibited by the supraquantized operator.  We can then confirm, once and for all, that the canonical commutation relation between time and energy influences the unitary dynamics of TOA operators.

\section{Acknowledgment}\label{sec:acknow}
D.A.L. Pablico and C.A.L. Arguelles are thankful for the scholarship support from the Department of Science and Technology - Science Education Institute (DOST-SEI).

\appendix
\section*{Appendix}\label{appendix1}
\setcounter{section}{1}
\subsection{Satisfying the time kernel equation}\label{app:finalchecktke}

As one final check, we now prove that Eq. (\ref{tsfull}) indeed satisfies the TKE. Taking the partial derivative of $T_S(u,v)$ with respect to $v$ gives
\begin{equation}
\begin{aligned}
\frac{\partial T_S(u,v)}{\partial v}&=\,G(v)\,\left(\frac{\mu}{2\hbar^2}\right)\int_{0}^{u} du'\,F(u')\frac{u'}{4}\\
&+ G(v)\,\left(\frac{\mu}{2\hbar^2}\right)^2 \int_{0}^{v} dv'\,G(v')\, \int_{0}^{u} du'\,F(u')\frac{u'}{4}\,\Tilde{F}(u,u')\,{}_0F_1 \left(;2;\left(\frac{\mu}{2\hbar^2}\right)\,\Tilde{G}(v,v')\,\Tilde{F}(u,u')\right).
\end{aligned}
\end{equation}
The above result is derived using the well-known Leibniz integral rule. Likewise, taking the partial derivative with respect to $u$ of the above equation leads to 
\begin{equation}\label{deri2tuv}
\begin{aligned}
&\frac{\partial^2 T_S(u,v)}{\partial u\,\partial v}=F(u)\,G(v)\,\frac{u}{4}\,\left(\frac{\mu}{2\hbar^2}\right)+F(u)\, G(v)\,\left(\frac{\mu}{2\hbar^2}\right)^2 \int_{0}^{v} dv'\,G(v')\, \int_{0}^{u} du'\,F(u')\frac{u'}{4}\\
&\times\left[{}_0F_1 \left(;2;\left(\frac{\mu}{2\hbar^2}\right)\,\Tilde{G}(v,v')\,\Tilde{F}(u,u')\right)+\,\frac{\mu}{2\hbar^2}\frac{\Tilde{G}(v,v')\,\Tilde{F}(u,u')}{2}{}_0F_1 \left(;3;\left(\frac{\mu}{2\hbar^2}\right)\,\Tilde{G}(v,v')\,\Tilde{F}(u,u')\right)\right].
\end{aligned}
\end{equation}

The factor in square brackets can be simplified using the well-known recurrence identity involving consecutive hypergeometric functions,
\begin{equation}\label{identityhypergeom}
	{}_0F_1 \left(;b;z\right) ={}_0F_1 \left(;b+1;z\right)+\frac{z}{b\,(b+1)}{}_0F_1 \left(;b+2;z\right). 
\end{equation}

Equation (\ref{deri2tuv}) further simplifies to 
\begin{equation}\label{deri2tuv2b}
\begin{aligned}
\frac{\partial^2 T_S(u,v)}{\partial u\,\partial v}=&F(u)\,G(v)\,\left(\frac{\mu}{2\hbar^2}\right)\\
&\times\left[\frac{u}{4}+\left(\frac{\mu}{2\hbar^2}\right)\int_{0}^{v} dv'\,G(v')\, \int_{0}^{u} du'\,F(u')\frac{u'}{4}  {}_0F_1 \left(;1;\left(\frac{\mu}{2\hbar^2}\right)\,\Tilde{G}(v,v')\,\Tilde{F}(u,u')\right) \right].
\end{aligned}
\end{equation}
Using our separability condition, $V(\left(u+v\right)/2)-V(\left(u-v\right)/2) = F(u)\,G(v)$ and noting that the factor in square brackets is exactly equal to the kernel factor $T_S(u,v)$ (\ref{tsfull}), we arrive at the partial differential equation
\begin{equation}
\frac{\partial^2 T_S(u,v)}{\partial u \partial v}=\frac{\mu}{2\hbar^2}\,\left(V\left(\frac{u+v}{2}\right)-V\left(\frac{u-v}{2}\right)\right)T_S(u,v),   
\end{equation}
which is exactly the TKE in its orginal form. This finally proves our exact-closed form solution of the TKE. 

\subsection{CPTOA operator under parity transformation}\label{subsec:cptoaparity}

In Ref. \cite{Galapon2018}, we have demonstrated that the quantized TOA operators possess eigenfunctions that have definite parities. Specifically, the eigenfunctions of an operator $\hat{\mathrm{T}}$ have definite parities if they themselves are also eigenfunctions of the parity operator $\hat{\Pi}$, whose action is defined by $\hat{\Pi}\, \psi(q)=\psi(-q)$. This happens when the operators $\hat{\mathrm{T}}$ and  $\hat{\Pi}$ commute. For a time operator of the form $(\hat{\mathrm{T}}\varphi)(q)=\int_{-\infty}^{\infty} \,dq' \,\langle q|\hat{\mathrm{T}}|q'\rangle\, \varphi(q')$, the desired equality $\hat{\mathrm{T}}\hat{\Pi}=\hat{\Pi}\hat{\mathrm{T}}$ is guaranteed if the time kernel $\langle q|\hat{\mathrm{T}}|q'\rangle$ satisfies the invariance under parity transformation in both of its arguments, that is, $\langle q|\hat{\mathrm{T}}|q'\rangle = \langle -q|\hat{\mathrm{T}}|-q'\rangle.$ Let us also check if the same property extends to the CPTOA operator. 

In the original $(q,q')$ coordinates, the kernel of the supraquantized TOA operator is
\begin{equation}\label{tsfull2qqp}
\begin{aligned}
\langle q|\hat{\mathrm{T}}_S|q'\rangle&=\frac{\mu}{i\hbar}\mathrm{sgn}(q-q')\,\Bigg[\frac{q+q'}{4}+\frac{\mu}{2\hbar^2}\int_{0}^{q-q'} dv'\,\, \int_{0}^{q+q'} du'\,\frac{u'}{4} \,\left[ V\left(\frac{u'+v'}{2}\right)-V\left(\frac{u'-v'}{2}\right)\right]\\
&\times{}_0F_1 \left(;1;\left(\frac{\mu}{2\hbar^2}\right)\,\int_{v'}^{q-q'}dv''\int_{u'}^{q+q'} du''\left[ V\left(\frac{u''+v''}{2}\right)-V\left(\frac{u''-v''}{2}\right)\right]\right)\Bigg].
\end{aligned}
\end{equation}

Taking the transformation $(q,q') \to (-q,-q')$ and performing the following change of variables: $(v' \to -v')$, $(u' \to -u')$, $(v'' \to -v'')$, and $(u'' \to -u'')$, we arrive at the following transformed kernel
\begin{equation}\label{tsfull2qqp2}
\begin{aligned}
\langle -q|\hat{\mathrm{T}}_S|-q'\rangle&=\frac{\mu}{i\hbar}\mathrm{sgn}(q-q')\Bigg[\frac{q+q'}{4}+\frac{\mu}{2\hbar^2}\int_{0}^{q-q'} dv'\int_{0}^{q+q'} du'\frac{u'}{4} \left[ V\left(-\frac{u'+v'}{2}\right)-V\left(-\frac{u'-v'}{2}\right)\right]\\
&\times{}_0F_1 \left(;1;\left(\frac{\mu}{2\hbar^2}\right)\,\int_{v'}^{q-q'}dv''\int_{u'}^{q+q'} du''\left[ V\left(-\frac{u''+v''}{2}\right)-V\left(-\frac{u''-v''}{2}\right)\right]\right)\Bigg].
	\end{aligned}
\end{equation}

Clearly, the desired invariance $\langle q|\hat{\mathrm{T}}_S|q'\rangle = \langle -q|\hat{\mathrm{T}}_S|-q'\rangle$ is preserved when the potential is even, $V(q)=V(-q)$. Therefore, the supraquantized TOA operator possesses odd and even eigenfunctions for an even interaction potential. 

For non-even potentials, the kernel is no longer invariant under parity transformation so that the CPTOA operator no longer commutes with the parity operator. Nevertheless, notice that the right-hand side of Eq. (\ref{tsfull2qqp2}) can also be interpreted as the time kernel factor of the supraquantized TOA operator for a given interaction potential $\hat{\Pi}V(q)=V(-q)$. Following Ref. \cite{Galapon2018}, we can define the specific conjugacy-preserving TOA operators $\hat{\mathrm{T}}^\pm_S$ where the positive sign corresponds to the potential $V=V(q)$ while the negative sign is for $V=V(-q)$. It can be shown that the eigenfunctions of the two operators are related to each other. Let us define $\tau$ and $\varphi^+(q)$ be the eigenvalue and eigenfunction, respectively, of the operator $\hat{\mathrm{T}}^+_S$. From the eigenvalue equation, $\hat{\mathrm{T}}^+_S\,\varphi^+(q)=\tau\,\varphi^+(q)$, one finds the relations
\begin{equation}
\tau \hat{\Pi}\varphi^+(q)=\int_{-\infty}^{\infty}dq'\langle -q|\hat{\mathrm{T}}^+_S|-q'\rangle \, \hat{\Pi}\, \varphi^+(q')=\int_{-\infty}^{\infty}dq'\langle q|\hat{\mathrm{T}}^-_S|q'\rangle \, \hat{\Pi}\, \varphi^+(q').
\end{equation}
The equality suggests that the eigenfunction $\varphi^-(q)=\hat{\Pi}\, \varphi^+(q)$ is an eigenfunction of the operator $\hat{\mathrm{T}}^-_S$ with the same eigenvalue $\tau$. The converse is also true. Defining $\varphi^-(q)$ to be an eigenfunction of  $\hat{\mathrm{T}}^-_S$ with the eigenvalue $\tau$, then $\varphi^+(q)=\hat{\Pi}\, \varphi^-(q)$ is an eigenfunction of $\hat{\mathrm{T}}^+_S$ for a given $\tau$.

\end{document}